\documentclass[aps,prx,longbibliography,a4paper,reprint,twocolumn,superscriptaddress,nofootinbib,floatfix]{revtex4-2}
\pdfoutput=1
\usepackage{amsmath}
\usepackage{amssymb}
\usepackage{dsfont}
\usepackage{verbatim}
\usepackage{graphicx}
\usepackage{tabularx}
\usepackage{xcolor}
\usepackage{mathtools}
\usepackage{bbold}
\usepackage{blkarray}
\usepackage{appendix}
\usepackage[shortlabels]{enumitem}
\usepackage{latexsym}
\usepackage{array}
\usepackage{colortbl}
\usepackage{bm}
\usepackage{amsthm}

\usepackage{tikz}
\usetikzlibrary{calc}
\usetikzlibrary{math}

\usepackage{thmtools}
\newtheorem{theorem}{Theorem}
\newtheorem{proposition}[theorem]{Proposition}
\usepackage{accents}

\newcommand{\ket}[1]{\mathinner{|#1\rangle}}

\newcommand{\ketbra}[2]{\mathinner{|#1\rangle\langle#2|}}

\DeclareMathOperator{\tr}{tr}

\newcommand{\CC}{\mathcal{C}}
\newcommand{\C}{\mathds{C}}
\newcommand{\ot}[0]{\otimes}
\renewcommand{\a}{\alpha}
\renewcommand{\b}{\beta}
\newcommand{\one}[0]{\mathds{1}}

\usepackage{dsfont} 

\usepackage{nicefrac}
\usepackage{hyperref}
    \hypersetup{
    colorlinks,
    linkcolor={red!40!black},
    citecolor={blue!50!black},
    urlcolor={blue!20!black}
}

\newcommand{\uibkexp}{\affiliation{Universit\"at Innsbruck, Institut f\"ur Experimentalphysik, Technikerstraße 25, Innsbruck, Austria}}
\newcommand{\iqoqi}{\affiliation{Institute for Quantum Optics and Quantum Information of the Austrian Academy of Sciences,  Technikerstraße 21a, Innsbruck, Austria}}
\newcommand{\aqt}{\affiliation{Alpine Quantum Technologies GmbH, Innsbruck, Austria}}
\newcommand{\sydney}{\affiliation{School of Mathematical and Physical Sciences, Macquarie University, NSW, Australia}}
\newcommand{\uq}{\affiliation{School of Mathematics and Physics, The University of Queensland, St Lucia, QLD, 4072, Australia}}

\usepackage{caption}
\usepackage{subcaption}
\usepackage{ragged2e}

\DeclareCaptionJustification{myjustified}{\justifying}
\begin{document}

\title{Holographic quantum codes with trapped ions}

\author{Alex Steiner}
\thanks{Contact: \href{mailto:alex.steiner@uibk.ac.at}{alex.steiner@uibk.ac.at}, \href{mailto:martin.ringbauer@uibk.ac.at}
{martin.ringbauer@uibk.ac.at}
}
\uibkexp
\author{Gerard Anglès Munné}
\affiliation{Institute of Theoretical Physics and Astrophysics, University of Gdańsk, 80-308 Gda\'nsk, Poland}
\author{Robert Freund}
\uibkexp
\author{Ivan Pogorelov}
\uibkexp
\author{Michael Meth}
\uibkexp

\author{Robert J. Harris}
\uq
\author{Gavin Brennen}
\sydney
\author{Thomas M. Stace}
\uq
\author{Thomas Monz}
\uibkexp
\aqt
\author{Rainer Blatt}
\uibkexp
\aqt
\iqoqi
\author{Felix Huber}
\affiliation{
Division of Quantum Information, 
Faculty of Mathematics, Physics and Informatics,
University of Gdańsk, Wita Stwosza 57, 80-308 Gdańsk, Poland}

\author{Martin Ringbauer}
\thanks{Contact: \href{mailto:alex.steiner@uibk.ac.at}{alex.steiner@uibk.ac.at}, \href{mailto:martin.ringbauer@uibk.ac.at}
{martin.ringbauer@uibk.ac.at}
}
\uibkexp

\begin{abstract}
Holography is a central concept at the intersection of gravity, condensed matter theory, and quantum information, linking the interior bulk of a system to its boundary. 
A model realizing key features of holographic systems is the 
holographic pentagon code by Pastawski et al.
Here we experimentally implement instances of the holographic pentagon and heptagon codes with trapped ions and test their properties: 
For the pentagon code, 
we recover logical bulk qubits from their nearby boundary and test the Ryu-Takayanagi entanglement area law.
For the heptagon code, we show that the transversal Hadamard gate native to the constituent Steane codes induces a single-qubit, correctable error in the holographic code.
Our implementation paves the way towards the use of holographic quantum codes for quantum information processing. 
\end{abstract}

\maketitle

Holography has emerged as a key concept in quantum gravity, condensed matter theory, and quantum information~\cite{Bousso_2002, Headrick_2007,RevModPhys.83.1057}. 
At its core, holography relates information to geometry. The key idea is that certain properties in the interior (bulk) of a volume are in one-to-one correspondence to properties on its lower-dimensional boundary. 
In other words, information read at the boundary allows one to infer information about the bulk.
This concept appears in quantum gravity, where the Anti-de Sitter/conformal field theory (AdS/CFT) correspondence links an Anti-de Sitter string theory in the bulk to a conformal field theory on the boundary~\cite{Maldacena1999}. A similar bulk-boundary correspondence exists 
in topological quantum matter, where topological invariants in the bulk are related to boundary (surface) states~\cite{Prodan_2016}.
As holography is concerned with recovering information of 
a quantum system from some of its subsystems, 
it was soon linked to quantum error correction (QEC)~\cite{
Verlinde_2013, Almheiri_2015, latorre2015holographiccodes}. 
Indeed, the task in QEC is similar: There, one encodes a state on $k$ logical qubits into the correlations of a larger number $n$ of physical qubits, so that errors affecting only a few physical qubits can be identified and corrected.

Recent work by Pastawski et al.~\cite{Pastawski2015} proposed a concrete quantum code realizing key features of holography. 
Based on the idea that tensor networks can capture key features of holography~\cite{Evenbly_2011, Almheiri_2015, latorre2015holographiccodes}, their model formulates a tensor network composed of quantum stabilizer codes. 
This so-called pentagon or Pastawski-Yoshida-Harlow-Preskill (HaPPY) code has its physical qubits located on the boundary of a hyperbolic tiling, 
whose interior (bulk) corresponds to encoded logical qubits.
Interestingly, this code shows certain holographic features: it allows one to recover a subset of logical (bulk) qubits from physical qubits on their nearby boundary, also known as partial recovery or bulk reconstruction, see Fig.~\ref{fig:partialrec}. It also satisfies the Ryu-Takayanagi formula~\cite{Ryu_2006}, which relates the amount of entanglement between boundary regions to the length of the partition of the bulk ($\gamma$ in Fig.~\ref{fig:partialrec}). 
This construction has been extended, including the heptagon code which has the Steane code as building blocks~\cite{Heptagon_code} (investigated in this work) and random tensor networks~\cite{Hayden_2016}.

Generally, when considering quantum codes, one is concerned with their distance $d$: this is the highest number, such that all errors affecting at most $d-1$ subsystems can be detected; this implies that all errors affecting at most $\lfloor \tfrac{d-1}{2} \rfloor$ subsystems can be corrected. 
Topological quantum codes such as the surface, toric, and color codes are additionally endowed with a topology (e.g.\ a torus). These codes are then also concerned with the requirement that all topologically trivial errors are correctable (that is, errors that are sufficiently local with respect to the topology). In particular, surface and color codes have been implemented in recent experiments \cite{Postler2022,Google2023,acharya2025,bluvstein2024,lacroix2024,zhaoRealizationErrorCorrectingSurface2021,postler2024,ryan-anderson2021}.

In contrast, holographic codes are not concerned with full recovery up to a given distance, but with partial recovery of quantum information from a subregion: 
That is, one does not require to correct all errors and reconstruct all encoded information fully, but only to reconstruct information residing in a bulk region from its nearby boundary. 
Surprisingly, this partial recovery is possible from different boundary regions. This corresponds to the fact that logical operators acting on the bulk can also be represented on multiple boundary regions, a feature also seen in quantum gravity~\cite{harlow2018tasilecturesemergencebulk}. Overall, by considering which boundary regions are able to reconstruct a given bulk region, one can attach a notion of a physical distance and locality to the code. This attaches a geometry to the code, which provides some intuition for how bulk information can be reconstructed from the boundary. Despite their theoretical interest, holographic quantum codes have not yet been implemented in the laboratory, a state of affairs this paper addresses.

We experimentally realize holographic QEC codes on a trapped-ion quantum processor to investigate their error correction capabilities and connection to the physics of holography through features such as partial recovery. 
Following the framework of Ref.~\cite{Angl_s_Munn__2024, PhysRevLett.127.040507, PhysRevA.105.052446}, the original tensor network code is understood as a stabilizer {\em graph code}, where logical states are described in terms of graph states. 
This allows us to prepare code states, perform logical operations, and implement decoding operations in an efficient manner. 
Experimentally, we demonstrate the key QEC properties of the code, including encoding, error detection, (transversal) logical operations, and (partial) decoding. We further study the physical properties, such as the bulk-boundary correspondence in partial decoding, and the entanglement structure. We conclude by discussing the role of holographic codes, not just for QEC, but also as an intriguing connection between quantum information and the physics of holography.

\section{Holographic quantum codes}\label{sec:holography}
 {\bf Theory. ---} A quantum error-correcting code aims to protect quantum information and computation against noise. 
This is achieved by encoding $k$ \emph{logical} qubits into a subspace of $n>k$ \emph{physical} qubits, in such a way that certain types of noise acting on the encoded information can be corrected by a recovery operation. 
A classic example is the repetition code, for which the states $\ket{0}$ and $\ket{1}$ 
are mapped to the two logical states
$(\ket{000} + \ket{111})\sqrt{2}$ 
and $(\ket{000} - \ket{111})\sqrt{2}$.
The parity measurements $ZZI$ and $IZZ$ (short for tensor product of Pauli operators i.e. $\sigma_z\otimes\sigma_z\otimes1$) then allow us to detect one single bit-flip error which maps $\ket{0} \longleftrightarrow \ket{1}$.

A holographic quantum code is a class of quantum codes that 
has an additional geometry attached to it~\footnote{While many codes have a notion of geometry attached to them, what most often matters is not their geometry but topology. For instance, in a surface or toric code all topologically trivial errors are correctable.}. 
Intuitively, the quantum information of the interior is mapped to the exterior, so that one is able to manipulate bulk information by acting on its boundary. 
Then, the logical qubits correspond to a bulk region $\mathcal{B} = (\mathds{C}^2)^{\otimes k}$, 
with the physical qubits corresponding
to its boundary $\partial \mathcal{B} = (\mathds{C}^2)^{\otimes n}$ (see Fig.~\ref{fig:MinimalGraph}(a)).
In this context, we denote the code subspace by $\CC \subset \partial \mathcal{B}$. 
What distinguishes a holographic code from other (e.g. topological or stabilizer) codes is the ability to also perform {\em partial} recovery (decoding) operations: any logical information corresponding to a subregion of the bulk can be recovered from its nearby boundary. 
This is in contrast with traditional quantum codes that are only concerned with the notion of distance and for which only a full recovery of the encoded information is possible.

\begin{figure}[!t]
    \centering
    \includegraphics[width=\linewidth]{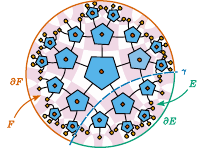} 
    \caption{\textbf{Holographic pentagon code.} The tensor network encodes $k$ bulk qubits (inner dots) into $n$ boundary qubits (outer dots). Each pentagon represents an absolutely maximally entangled state of 6 qubits (center plus five legs). 
    For particular cuts $\gamma$, the geometry of the tensor network allows for partial decoding, i.e. one recovers the qubits from the bulk region $E$ by reading only the qubits from its nearby boundary $\partial E$.
    \label{fig:partialrec}
    }
\end{figure}

The holographic code we consider here 
is given as a stabilizer code.
The latter is described by a commuting set of independent elements $g_1, \dots, g_{n-k}$ of the $n$-qubit Pauli group that generates the set of stabilizers $S = \langle g_1, \dots, g_{n-k}\rangle$, where $-\one \not \in S$~\cite{gottesman1997stabilizercodesquantumerror}. 
The code space is characterized as the joint $+1$ eigenspace of all stabilizer elements $s \in S$.
Therefore, the projector onto the code space $\Pi$ satisfies $s \Pi = \Pi$ and it is given by $\Pi = \prod_{i=1}^{n-k} (\one + g_i)/2 = \sum_{s\in S} s / 2^{n-k}$. 
Errors for stabilizer codes can be detected in the following way:
a stabilizer element $s$ is able to detect all errors that anti-commute with $s$. 
Measuring $s$ then acts as a syndrome measurement, indicating which errors have occurred: If a syndrome measurement returns $\langle s_i \rangle = +1$ for all $i=1, \dots, n-k$, then no (detectable) error has occurred, while $\langle s_i \rangle = -1$ for some $s_i$ indicates an error. 

Any stabilizer code encoding $k$ to $n$ qubits can be represented as a graph $G$ on the $(n+k)$ vertices, which in turn corresponds to a graph state $\ket{G}$. 
This is known as a graph code and $\ket{G}$ is defined as
\begin{equation}\label{eq:graphcz}
    \ket{G}=\prod_{(u,v)\in E} \text{CZ}_{uv} \ket{+}^{\otimes {n+k}}\,,
\end{equation}
with $E$ being the edges of $G$ and $\text{CZ}_{uv}$ being a controlled-$Z$ gate acting on qubits $u$ and $v$. 
The graph form allows us to extract the \emph{logical states} $\ket{G_{i_1\cdots i_k}}$, which form a basis in the code subspace $\CC$ and are, up to phases, $n$-qubit graph states.
In addition, the logical operators $\bar{X}_j$ and $\bar{Z}_j$ can also be derived from $G$.
These operators are defined to act on $\ket{G_{i_1\cdots i_k}}\in \CC$ in the same way as Pauli operators $X_j$ and $Z_j$ act on a computational basis $\ket{i_1 \cdots i_k}\in \mathcal{B}$.
The encoding and decoding operations are thus contained in the graph structure.

The task of encoding information in a quantum code can be represented by an isometry $\mathcal{E}: (\C^2)^{\ot k} \to (\C^2)^{\ot n}$, such that for a given noisy channel $\mathcal{N}$ acting on the physical qubits, there exists a recovery operation $\mathcal{R}$ satisfying the composition of channels
\begin{equation}
    \mathcal{R} \circ \mathcal{N} \circ \mathcal{E}(\varrho) = \varrho 
\end{equation}
for every $\varrho$ on $(\C^2)^{\ot k}$.
In the case of a holographic code, we consider certain cuts $\gamma$ that split the bulk region into two subregions $E$ and $F$, and the boundary into $\partial E$ and $\partial F$ (see Fig.~\ref{fig:partialrec}). Then, one additionally has the property
\begin{equation}\label{eq:holo_rec}
    \big(
    (\mathcal{\widetilde R} \ot \tr_{\partial F})
    \circ
    \mathcal{\widetilde{N}}
   \circ 
    \mathcal{E}(\varrho)\big) 
    = \varrho_E \,,
\end{equation}
where $\varrho_E = \tr_F(\varrho)$. 
Here $\mathcal{E}: {\mathcal{B} \to \partial \mathcal{B}}$ encodes bulk to boundary, $\mathcal{\widetilde{N}}: \partial\mathcal{B}  \to \partial\mathcal{B}$ is noise acting non-trivially on a part $\partial F$ of the boundary only, and $\mathcal{\widetilde R}: \partial E \to E$ is a {\em partial} recovery operation. That is, the quantum information which originated in a bulk region $E$ can be recovered from its nearby boundary region $\partial E$.
This allows for holographic codes to recover only a part of the encoded information (here: the part which was encoded from $E$ into $\partial E$),
without the need to decode the full code.
In contrast, a traditional code would require access to the full system $\partial E \cup \partial F$ even if only some logical qubits need to be recovered.

The regions for which Eq.~\eqref{eq:holo_rec} holds are determined by the specific holographic code. In the following, we give examples of this property by showing selected regions for which recovery is possible (see Fig.~\ref{fig:partialrec}). Details on the unitary gates and circuits corresponding to these maps can be found in the supplemental material~\ref{ss:Additional_logical_operators_HaPPY}.

Holographic systems also satisfy the Ryu-Takayanagi (RT) formula~\cite{Ryu_2006, Harlow_2017,harlow2018tasilecturesemergencebulk}.  
For the holographic code considered here, the RT formula states that the von Neumann entropy of the reduced state $\sigma_{\partial E}$ on $\partial E$ of any $\ket{G_{i_1\cdots i_k}}$ satisfies
\begin{align}\label{eq:RTformula}
S(\partial E)=- \tr\left(\sigma_{\partial E} \log_2 \sigma_{\partial E}\right) =|\gamma| \,,
\end{align}
where $|\gamma|$ is the number of contracted indices in the tensor network between $E$ and $F$.
That is, the entropy of any state defined on the boundary $\partial E$ is proportional to the length of the cut $\gamma$ which separates $E$ from $F$.
Importantly, the RT formula only holds for cuts $\gamma$ for which complementary recovery is possible, i.e., for which there exists an isometry mapping $E$ to $\partial E$ and another $F$ to $\partial F$~\cite[Section 4.2]{Harlow_2017}.

 {\bf Pentagon code. ---}
We implement the minimal instance of the hyperbolic pentagon code~\cite{Pastawski2015}, consisting of $12$ physical boundary qubits and $4$ logical bulk degrees of freedom (labeled A-D), as shown in Fig.~\ref{fig:MinimalGraph}(a) and (b). 
We use the formulation in terms of its corresponding stabilizer graph code, derived by the authors in Ref.~\cite{Angl_s_Munn__2024}. The pentagon code can be represented by a stabilizer state generated by contracting four six-qubit absolutely maximally entangled (AME) states~\cite{Pastawski2015}. Here, contraction means projecting corresponding qubits of different AME states onto a maximally entangled state, thereby linking the independent AME states into a single stabilizer state. Formally, this corresponds to a tensor-index contraction~\cite{Hayden_2016}, which is equivalent to projecting onto a Bell state and performing a partial trace over the contracted qubits. Following Ref.~\cite{Angl_s_Munn__2024}, the resulting stabilizer state can be mapped to a graph code.
The logical zero state of the resultant graph code is the graph state $\ket{G_{0000}}$ illustrated in Fig.~\ref{fig:MinimalGraph}(b).
Preparing $\ket{G_{0000}}$ then requires 12 qubits initialized in the $\ket{+}$ state, followed by 28 $\text{CZ}$-gates. Details on the state preparation, as well as the stabilizer generators and the logical operators, are given in Supp. Mat.~\ref{ss:Stabilizer_generators_logical_operators_HAPPY}.

The main focus of the pentagon code is its implementation on a trapped-ion quantum computer and the examination of its holographic features. The success of the implementation is evaluated using two complementary approaches. We assess state preparation by measuring all stabilizer generators and logical operators of various logical states, as well as performing direct fidelity estimation to find a lower limit to the state preparation fidelity. We place particular emphasis on the partial decoding of logical information. Specifically, we study the recovery of the encoded information under the application of errors on different qubits on the boundary region and measure the corresponding von Neumann entropy. In addition, we also demonstrate the ability to detect single-qubit errors on the boundary. Specifically, we show that each single-qubit error produces a unique syndrome.

{\bf Heptagon code. ---}
An alternative approach to QEC using holographic codes is provided by the holographic heptagon code. This code belongs to the Calderbank-Shor-Steane (CSS) class of stabilizer codes~\cite{Heptagon_code}, allowing all stabilizers and logical operators consisting only of $X$ or $Z$ Pauli operators to be read out simultaneously. As a result, it makes more efficient use, compared to the pentagon code, of the mid-circuit resources available on current quantum hardware platforms. 
We demonstrate the implementation of the smallest instance of the holographic heptagon code~\cite{Heptagon_code}, in which two logical qubits are encoded into 12 physical qubits. The holographic heptagon code can be created by contracting two $[\![7, 1, 3]\!]$ Steane codes~\cite{Heptagon_code}. Here, contraction denotes the projection of corresponding qubits onto a Bell state, followed by a partial trace over those qubits.
Details on the stabilizer generators and logical operators as well as on the sate preparation are given in Supp. Mat.~\ref{SS:Init_heptagon}.

Our main focus in the following is the implementation of transversal single-qubit and two-qubit gates. A transversal gate does not connect any physical qubits inside the same logical qubit with multi-qubit gates.  While the Steane code admits a transversal Hadamard gate, this does not hold for the heptagon holographic code, since it is no longer a block code. However, we show that applying a physical Hadamard to all physical qubits corresponding to a logical qubit realizes a logical Hadamard gate on the respective logical qubit up to a single-qubit error on the neighboring logical qubits in the code, which we show how to correct in post-processing.

\section{Trapped-ion implementation}
All experiments are carried out on a trapped-ion quantum processor, as described in Ref.~\cite{pogorelov_compact_ion_2021}. Sixteen $^{40}\text{Ca}^+$ ions are trapped in a macroscopic linear Paul trap, where the electronic states are controlled via laser pulses. Each ion encodes one qubit into the electronic Zeeman levels $\ket{0} = \ket{4^2\text{S}_{1/2}, m_J=-1/2}$ and $\ket{1} = \ket{3^2\text{D}_{5/2}, m_J=-1/2}$, which are coupled via an optical quadrupole transition at a wavelength of 729\,nm. Coulomb interaction between the ions gives rise to collective motional modes, which are used to mediate entangling operations between any desired pair of qubits. The capability to selectively address individual as well as pairs of ions enables arbitrary single-qudit and two-qudit interactions. The native gate set consists of arbitrary rotations around the  X and Y axes of the Bloch sphere, virtual ${Z}$-rotations implemented in software~\cite{mckay_efficient_z_2017}, and the maximally entangling two-qubit M\o{}lmer-S\o{}rensen (MS, or $R_{XX}$) gate $\exp (\texttt{i}\frac{\pi}{4}X\otimes X)$~\cite{sorensen_quantum_computation_1999}, which is equivalent to a $\text{CNOT}$ gate up to local rotations~\cite{maslov_basic_circuit_2017}.

The string of 16 ions is divided into two groups, and used for different purposes: four ions in the center of the trap are idling and do not take part in any computation, while the logical qubits are encoded into the six ions at the ends of the ion chain. This type of implementation reduces cross-talk between neighboring ions and improves two-qubit gate fidelity since the ion distance in the string decreases towards the middle of the string. Moreover, cooling the highest-order modes, affecting predominantly the central ions, to the motional ground state is challenging. Hot modes together with imperfect calibration reduce the fidelity of the entangling gates.

{\subsection{Pentagon Code}}\label{subsection_happy_encoding}
The experimental preparation of the graph state $\ket{G^{}_{0000}}$ is evaluated in two ways: Estimating the fidelity by measuring the stabilizer generators of one logical state, and performing Direct Fidelity Estimation~(DFE)~\cite{DirectFidelityEstimation}.

\begin{figure}
	\centering
	
	\begin{subfigure}[t]{0.475\linewidth}
		\captionsetup{justification=raggedright,
			singlelinecheck=false,
			position=top}
		\caption{{\small a}}
		\includegraphics[width=\linewidth]{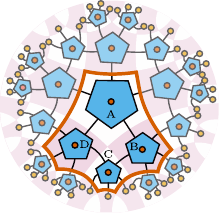}
	\end{subfigure}
	\quad
    \begin{subfigure}[t]{0.475\linewidth}
		\captionsetup{justification=raggedright,
			singlelinecheck=false,
			position=top}
		\caption{{\small b}}
		\includegraphics[width=\linewidth]{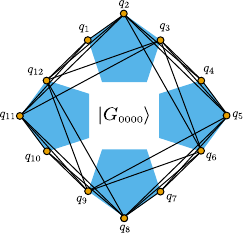}
	\end{subfigure}
	\quad
	\begin{subfigure}[t]{\linewidth}
		\captionsetup{justification=raggedright,
			singlelinecheck=false,
			position=top}
		\caption{{\small c}}
		\includegraphics[width=\linewidth]{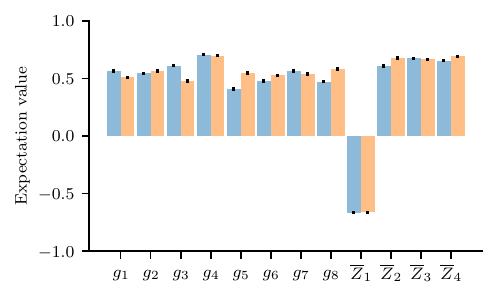}
	\end{subfigure}
	 \caption{
        \textbf{A holographic graph code.}
        \textbf{(a)} Highlights a small instance of the pentagon code, which we prepare experimentally. The four logical bulk qubits are labeled by A,B,C,D. \textbf{(b)} The logical-zero graph state $\ket{G_{0000}}$ of the code instance in (a), after being mapped to a graph code. Recall that $\ket{G_{0000}}$ is the boundary state that results from encoding the bulk state $\ket{0000}$. Hence, the bulk qubits are collectively encoded in the state and cannot be associated with specific boundary qubits.
        \textbf{(c)} The simulated (orange) and measured (blue) expectation values of the stabilizer generators and logical $Z-$operators of the initialization of the pentagon code in the $\left|1000\right>$ state. The error bars are computed from quantum projection noise and represent one standard deviation of statistical uncertainty. Note that the simulations are based on a simplified depolarizing noise model, and some deviations from the measured values are thus expected.
        }
     \label{fig:MinimalGraph}
\end{figure}

For the first method the expectation values of the stabilizer generators and the logical operators of $\ket{G^{}_{1000}}$ are measured as shown in Fig.~\ref{fig:MinimalGraph}(c). The experimental data is compared with a numerical simulation incorporating a depolarization channel to model the noise introduced by both single-qubit and two-qubit gate operations. However, modeling the system with a depolarizing channel as the sole noise source constitutes a simplification. Consequently, deviations between experimental results and numerical simulations are to be expected. A comprehensive description of the simulation method is presented in Supp.\ Mat.~\ref{SS:Simulation_methods}. The expectation values for the pentagon code are measured sequentially on a shot-by-shot basis, as the combined number of physical qubits of the code and ancilla qubits required for mapping the stabilizer generators exceeds the available quantum computational resources. Therefore, no stabilizer mapping on ancillas was performed and the stabilizer generators are measured directly on the physical qubits. Each stabilizer generator and logical operator expectation value is averaged over the outcomes of $2 \cdot 10^4$ shots. The average expectation value of all stabilizer generators and logical operators is 0.58(2) [simulation: 0.59(2)].

The second approach is performing Direct Fidelity Estimation on $\ket{G^{}_{0000}}$ to determine the fidelity from a limited number of measurements. Since $\ket{G^{}_{0000}}$ is a stabilizer state, it is sufficient to measure the expectation value of a small random subset of its stabilizer group to estimate the fidelity between the desired state $\rho$ and the target state $\sigma$~\cite{ringbauer2025}. 
We randomly select 450 stabilizers from the full stabilizer group and measure their expectation values taking 100 shots for each outcome. The measured fidelity of the logical zero state is estimated to be $F(\rho, \sigma)=0.309(13)$, which is well above the threshold of $1/12$ for the maximally mixed state. Fidelity uncertainty is estimated via sub-sampling, using $100$ subgroups of $100$ randomly selected expectation values each. It is calculated as the mean of the standard deviations across these subgroups. We emphasize that the fidelity, as a global property of the state, is expected to be more sensitive to noise than the local stabilizer expectation values and logical operators quoted above. 

{\bf Logical operations. ---}
\label{subsection_happy_gates}
We demonstrate the ability to perform global and addressed gates on the logical qubits A-D. Therefore, we prepare the state $\ket{G^{}_{0000}}$ and realize the logical single-qubit $\overline{X}$ and $\overline{H}$, as well as a logical $\overline{\text{CZ}}$ gate. Details regarding each gate are presented in Supp.\ Mat.~\ref{ss:Additional_logical_operators_HaPPY}.

The successful implementation of each gate is evaluated by measuring the expectation values of the stabilizer generators and logical operators. Each outcome is compared to numerical simulations using the simple depolarizing model discussed in the supplementary material. We find the average expectation values of the stabilizer and logical operators to be:
\vspace{0.2cm}
\begin{enumerate}[(i)]
    \item 0.58(2) [sim:\,0.59(2)] for $\ket{G^{}_{1000}}$, 
    \item 0.51(2) [sim:\,0.58(2)] for $\ket{G^{}_{+000}}$,
    \item  0.47(2) [sim:\,0.53(2)] for $\tfrac{1}{\sqrt{2}}\left(\ket{G^{}_{+0++} + G^{}_{-1++}}\right)$\,.
\end{enumerate}
The expectation values of all stabilizers and logical operators are discussed in the Supp.\ Mat. ~\ref{SS:Eigenvalues_gates_HAPPY} and illustrated in Fig.~\ref{fig:SS_eigenvalues_happy_init}(a-d).

{\bf Error correction. ---}\label{subsection_happy_syndrome_readout}
Here we investigate the ability to perform standard error-correction for the pentagon code, disregarding its holographic features.
The small instance considered in this work is formally a $[\![12,4,3]\!]$ quantum code (see Fig.~\ref{fig:MinimalGraph}(a)), 
and thus can correct all single-qubit errors.
This encodes $4$ logical qubits into $12$ physical qubits, with a code distance $3$, thus being able to correct $1$ errors.
This can be seen by noting that all bulk qubits can be recovered from up to $2$-qubit erasures, implying that the code distance is $3$.
Identification of such errors is a crucial step towards realizing quantum error correction. To this end, we initialize the pentagon code in the logical zero state $\ket{G^{}_{0000}}$. An artificial $X$-, $Y$-, or $Z$-type error is then applied to a single physical qubit, followed by a destructive measurement of one of the stabilizer generators $g_1$ through $g_8$. The resulting syndromes are compared to both numerical simulations and the expected ideal outcomes. A lookup table mapping each error type to its corresponding syndrome is provided in the Supp.\ Mat.\ ~\ref{SS:Lookup_table_HaPPY}.

We demonstrate error detection for 6 randomly chosen single-qubit errors, out of the complete set of 36. Each selected syndrome is measured over $2\cdot10^4$ shots. The averaged expectation values for all stabilizer generators are presented in Fig.~\ref{fig:Syndrome}. All measured syndromes agree with the expected outcomes after introducing the corresponding error, demonstrating the code's ability to detect and correct them.

{\bf Partial decoding. ---}\label{subsection_happy_partial_decoding}
A key feature of holographic codes is partial decoding. Logical information encoded in the bulk can be recovered via a set of qubits on the boundary, due to the bulk-boundary correspondence. Since this boundary region is not unique, several decoding strategies can be employed to recover a logical qubit. We demonstrate the decoding of two, three, and all four logical qubits using the minimal number of boundary qubits, see Supp.\ Mat.~\ref{sect:enc_extra_q} and \ref{SS:alternative_partial_decoding}.
While we present only a subset of possible decoding circuits for recovery, the decoding strategy is symmetric under reflections and rotations. The circuits shown, therefore, represent distinct decoding geometries up to symmetry and are valid for any reflected or rotated configuration. 
All decoding strategies follow a similar scheme: 
The code is partitioned (cut) into two regions, with one region containing the boundary region required for decoding the selected bulk qubits; its complement can be ignored. The bulk information is then mapped onto the boundary by an appropriate decoding circuit, shown in Fig.~\ref{fig:unitary_rec} (cut 1), ~\ref{fig:D3} (cut 2), and ~\ref{fig:D4} (cut 3) for decoding two, three, and all four bulk qubits, respectively. 

The success of partial decoding is evaluated by measuring the decoded logical qubits in the computational basis after the circuits are applied.
The first experiment aims to decode two logical qubits, namely logical qubits A and B, as illustrated in Fig.~\ref{fig:unitary_rec}(a).
Here, the operations required for partial decoding are applied solely to boundary qubits 1-5, highlighting that only a confined region is required to decode the bulk information of logical qubits A and B. Furthermore, we investigate the decoding of three logical qubits A, B, D (cut 2) according to Fig.~\ref{fig:D3} and of all four logical qubits A, B, C, D (cut 3), according to Fig.~\ref{fig:D4}. The expectation values of the decoded logical qubits for cuts 1-3 are shown in Fig.~\ref{fig:exp_partial_decoding}(a-c) and compared to simulations, with the numerical values given in Supp.\ Mat.~\ref{SS:Partial_decoding_2_3}.
The experimentally obtained expectation values are in good agreement with the numerical simulations employing a depolarizing noise model, 
confirming that the decoding protocol reliably recovers the encoded logical qubits from the respective boundary. The variation in the expectation values between the three cases agrees with the numerical simulation, and is thus likely due to the different numbers of entangling gates in the corresponding decoding circuits.

In the case of decoding 4 logical qubits, the procedure corresponds to that of standard quantum error correction. 
Interestingly, the holographic code is able to correct more errors than its constituent codes:
For example, holographic codes have the potential to correct so-called burst errors, which are simultaneous errors on neighboring qubits~\cite{fan2025characterizingbursterrorcorrection,doi:10.1143/JPSJ.69.3540}. Fig.~\ref{fig:D4} shows an example of burst error correction for qubits $q_5, q_6, q_7$, which goes beyond the capability of the constituent $[\![5,1,3]\!]$ codes.
\begin{figure}
    \centering
	\hspace{-1.5cm}
	\begin{subfigure}[t]{0.45\linewidth}\captionsetup{justification=raggedright,
			singlelinecheck=false,
			position=top}
		\caption{{\small a}}
		\vspace*{-0.3cm}\includegraphics[width=\linewidth]{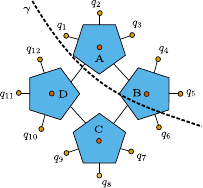}
	\end{subfigure}
	\hspace{0.1cm}
	\begin{subfigure}[t]{0.3\linewidth}
	\captionsetup{justification=raggedright,
			singlelinecheck=false,
			position=top}
		\caption{{\small b}}
        \hspace*{14pt}
        \includegraphics[width=\linewidth]{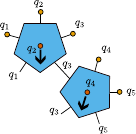}
	\end{subfigure}
	
    \begin{subfigure}[t]{0.45\linewidth}
		\captionsetup{justification=raggedright,
			singlelinecheck=false,
			position=top}
		\caption{{\small c}}
    		\vspace*{-0.2cm}\includegraphics[width=\linewidth]{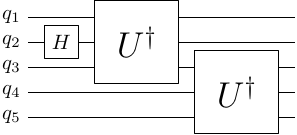}
	\end{subfigure}
	\quad
    \begin{subfigure}[t]{0.48\linewidth}
		\captionsetup{justification=raggedright,
			singlelinecheck=false,
			position=top}
		\caption{{\small d}}
    		\vspace*{-0.2cm}
            \includegraphics[width=\linewidth]{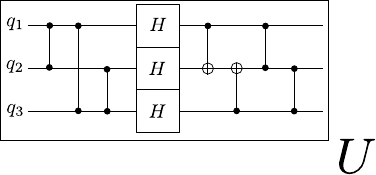}
	\end{subfigure}
    \caption{\textbf{Cut 1: Partial decoding of two bulk qubits.} 
    \textbf{(a)} A possible cut $\gamma$ allowing for partial decoding using the code instance from Fig.~\ref{fig:MinimalGraph}(a).
    \textbf{(b)} The information of two bulk qubits A and B can be recovered from the five boundary qubits $q_1$ to $q_5$.
    %
    This recovery is performed via a unitary operator constructed from the concatenation of two AME states as illustrated in (b).
    The arrows indicate the direction of the unitary, showing that qubits $q_2$ and $q_4$ will contain the information of the bulk qubits A and B after the operator is applied.
    \textbf{(c)} Recovery unitary circuit derived from the concatenation of the two AME states in (b), where each AME state corresponds to a unitary $U$.
    An additional Hadamard gate $H$ is applied to $q_3$.
    This is because we are performing partial decoding using the holographic graph code of Fig.~\ref{fig:MinimalGraph}(b), which is, up to local Clifford, equivalent to the holographic code obtained from a direct contraction of the tensor network shown in Fig.~\ref{fig:MinimalGraph}(a).
    More details of the derivation can be found in Supp.\ Mat.\ ~\ref{SS:alternative_partial_decoding}.
    \textbf{(d)} Decomposition of the unitary $U$, as detailed in the Supp. Mat. ~\ref{sect:enc_extra_q}. 
	The $H$ box represents a Hadamard gate, the dot-line-dot a CZ gate, and the dot-line-$\oplus$ is a CNOT gate.
    All gates to perform partial decoding are extracted from the graph state form.
    }
    \label{fig:unitary_rec}
\end{figure}

\begin{figure}
    \centering
    \begin{subfigure}[t]{0.45\linewidth}
		\captionsetup{justification=raggedright,
			singlelinecheck=false,
			position=top}
		\caption{{\small a}}
		\includegraphics[width=\linewidth]{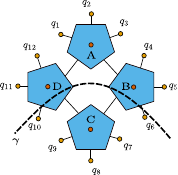}
	\end{subfigure}
	\quad
	\begin{subfigure}[t]{0.45\linewidth}
		\captionsetup{justification=raggedright,
			singlelinecheck=false,
			position=top}
		\caption{{\small b}}
		\includegraphics[width=\linewidth]{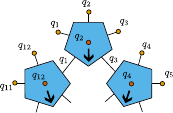}
	\end{subfigure}
	\quad
	\begin{subfigure}[t]{0.95\linewidth}
		\captionsetup{justification=raggedright,
			singlelinecheck=false,
			position=top}
		\caption{{\small c}}
		\includegraphics[width=\linewidth]{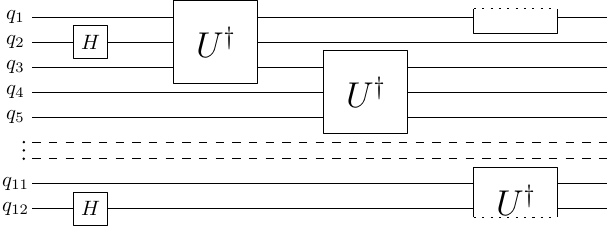}
	\end{subfigure}
    \caption{\textbf{Cut 2: Partial decoding of three bulk qubits.} \textbf{(a)} From the cut $\gamma$, one can recover three bulk qubits from seven boundary qubits ($q_1$, $q_2$, $q_3$ $q_4$,  $q_5$, $q_{11}$ and $q_{12}$). 
    \textbf{(b)} The recovery can be done via a unitary operation that can be expressed as a concatenation of three AME states.
    The arrows indicate the direction of the unitary, showing that qubits $q_2$, $q_4$, and $q_{12}$ will contain the information of the bulk qubits A, B, and D after the operator is applied.
    \textbf{(c)} The recovery operator is decomposed in the following circuit, where the box $U^\dagger$ corresponds to the reverse of Fig.~\ref{fig:unitary_rec}(d).
    Additional Hadamard gates $H$ are needed to perform the partial decoding, since we are using the graph code from \cite{Angl_s_Munn__2024} for the encoding.
    After applying the partial decoding circuit, the qubits $q_2$, $q_4$, and $q_{12}$ contain the information of the three bulk qubits.
    }
    \label{fig:D3}
\end{figure}

\begin{figure}
    \centering
    \begin{subfigure}[t]{0.45\linewidth}
		\captionsetup{justification=raggedright,
			singlelinecheck=false,
			position=top}
		\caption{{\small a}}
		\includegraphics[width=\linewidth]{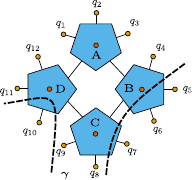}
	\end{subfigure}
	\quad
	\begin{subfigure}[t]{0.45\linewidth}
		\captionsetup{justification=raggedright,
			singlelinecheck=false,
			position=top}
		\caption{{\small b}}
		\includegraphics[width=\linewidth]{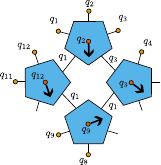}
	\end{subfigure}
	\quad
	\begin{subfigure}[t]{0.95\linewidth}
		\captionsetup{justification=raggedright,
			singlelinecheck=false,
			position=top}
		\caption{{\small c}}
		\includegraphics[width=\linewidth]{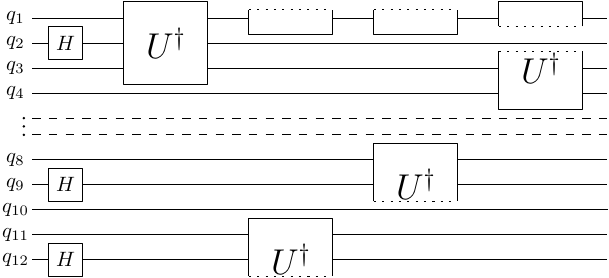}
	\end{subfigure}
    \caption{\textbf{Cut 3: Partial decoding of four bulk qubits.}
    \textbf{(a)} The cut $\gamma$ allows one to recover four bulk qubits from eight boundary qubits ($q_1$, $q_2$, $q_3$ $q_4$,  $q_8$, $q_9$, $q_{11}$ and $q_{12}$). 
    \textbf{(b)} The recovery unitary can be expressed as a concatenation of four AME states as shown in Fig.~\ref{fig:D4}(b). The arrows show the direction of the unitary, indicating that qubits $q_2$, $q_3$, $q_9$ and $q_{12}$ will contain the information of the bulk qubits A, B, C and D after the operator is applied.
    \textbf{(c)} The unitary is decomposed in the following circuit, where the box $U^\dagger$ corresponds to the reverse of Fig.~\ref{fig:unitary_rec}(d). Additional Hadamard gates $H$ are needed to perform partial decoding since we used the the graph code derived in~\cite{Angl_s_Munn__2024} for the encoding. After applying the partial decoding circuit, the qubits $q_2$, $q_3$, $q_9$ and $q_{12}$ contain the information of the four bulk qubits.
    }
    \label{fig:D4}
\end{figure}

Next we investigate the effect of introducing an error on a region of the boundary not involved in any of the partial decoding procedures. Specifically, an entangling MS gate is applied between boundary qubits 6 and 10 prior to decoding. The measured expectation values for all three cases, shown in Fig.~\ref{fig:exp_partial_decoding}(g-i), remain consistent with those obtained in the absence of this additional operation. This is to be expected, since the operation was applied outside of the region involved in the decoding for any considered cut. 

\begin{figure}
    \centering\includegraphics[width=\linewidth]{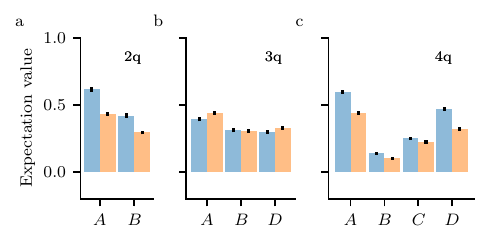}
    \includegraphics[width=\linewidth]{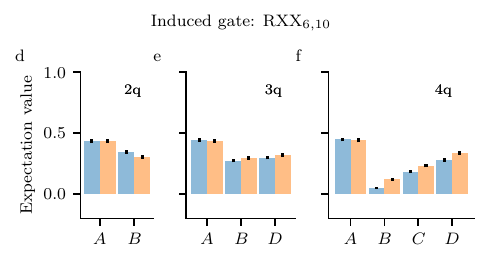}
    \includegraphics[width=\linewidth]{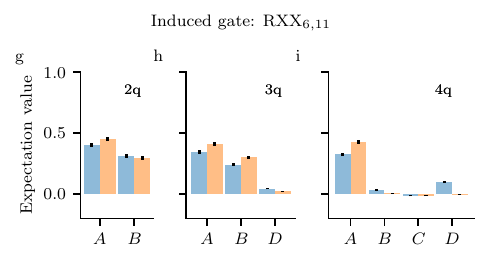}
    \caption{\textbf{Partial decoding.} The code is initialized in $\ket{G^{}_{0000}}$ for all experiments. Simulated (orange) and experimentally measured (blue) $Z$-expectation value of the decoded bulk qubits A-D are presented for the implemented partial decoding operations 1-3 (cut 1-3).  The error bars are computed from quantum projection noise and represent one standard deviation of statistical uncertainty. \textbf{(a)-(c)} show results for the decoding of 2 (cut 1), 3 (cut 2), and 4 (cut 3) bulk qubits, respectively. \textbf{(d)-(f)} display the expectation values obtained when an additional MS gate is applied between boundary qubits 6 and 10. As this entangling gate does not act on the qubits involved in the decoding process, the resulting expectation values remain consistent with those in (a)-(c). In contrast, \textbf{(g)-(i)} show the expectation values when an additional MS gate is applied between boundary qubits 6 and 11. This operation entangles qubits of the two regions separated by the cut and therefore, partially prevent the recovery of the logical qubits.}
    \label{fig:exp_partial_decoding}
\end{figure}

In contrast, if a boundary qubit used in the partial decoding process is corrupted or lost from the computational subspace, the logical information may not be recoverable. In this case, we apply an MS gate between qubits 6 and 11, thereby introducing a non recoverable error affecting a boundary qubit used for partial decoding. While neither qubit 6 nor 11 is used in decoding logical qubits A and B, qubit 11 is required for decoding qubit D in the three-qubit decoding case and for decoding qubit B-D in the four qubit case. As shown in Fig.~\ref{fig:exp_partial_decoding}(g), logical qubits A and B remain unaffected and can be decoded reliably when decoding two logical qubits from the boundary (cut 1). Since the MS gate was applied at a boundary qubit not involved in the decoding, this outcome is expected and consistent with previous measurements. When decoding three logical qubits, qubit D can no longer be recovered, while the result for qubits A and B remain consistent with the previous data. This is to be expected, as decoding qubits A and B requires only boundary qubits 1-5,  whereas decoding qubit D requires boundary qubits 1, 11, and 12. A similar behavior is observed when decoding all four qubits A-D: we are able to decode qubit A, as none of its required boundary qubits are affected by the MS gate at the time of decoding. However, qubits B-D cannot be decoded, as their decoding circuits involve boundary qubit 11.
We conclude that we are able to decode the information of those logical qubits, for which the boundary qubits necessary for decoding are not corrupted. 
%
%

{\bf Entropy measurement. ---}
Thus far, we investigated the success of partial decoding by measuring the expectation values of the decoded bulk qubits. We measure the entropy of the boundary belonging to one side of the cut for an encoded product state, in order to verify the Ryu-Takayanagi formula~\eqref{eq:RTformula}.
The code is initialized in $\ket{G^{}_{0000}}$, and state tomography is performed on the corresponding boundary qubits of cuts 1-3 or their complement if that is smaller. The density matrix is extracted from the measurements by performing a maximum likelihood estimation~\cite{MLE_state_tomo_hradil, MLE_state_tomo_banaszek} to exclude unphysical results. We calculate the von Neumann entropy $S$ from the measured density matrix. The entropy for decoding 2, 3, and 4 bulk qubits is 3.96(2) [Ideal: 3], 4.13(2) [Ideal: 4], and 3.818(10) [Ideal: 4], respectively. The ideal entropy is given by the number of contracted indices from the cut $\gamma$ [see Eq.~\eqref{eq:RTformula}]. The density matrices of all three cuts are displayed in Fig.~\ref{fig:Tomography_data_all_cuts} in Supp. Mat.~\ref{SS:Partial_decoding_tomography}. We note that the measured entropy deviates more strongly from the ideal case for the first cut. We attribute this to the finite upper bound in entropy given by the structure of the problem. The first cut is expected to be significantly further from the maximum and is thus relatively more susceptible to experimental imperfections in the preparation of the logical zero state $\ket{G_{0000}}$ than the other two cuts.
The measurement procedure for the second and third cut differs from the first: we compute the entropy of the complement of the decoding subset rather than that of the subset itself. Consequently, the number of qubits involved in state tomography is reduced from 7 to 5 and from 8 to 4 for cuts 2 and 3, respectively. The expected outcome of the state tomography of cut 3 corresponds to a maximally mixed state, since it is carried out on the complement of the decoding subset. Imperfections in the entanglement lead to residual pure-state contributions, reducing the measured entropy with respect to the ideal case.

\label{sec:encoding}

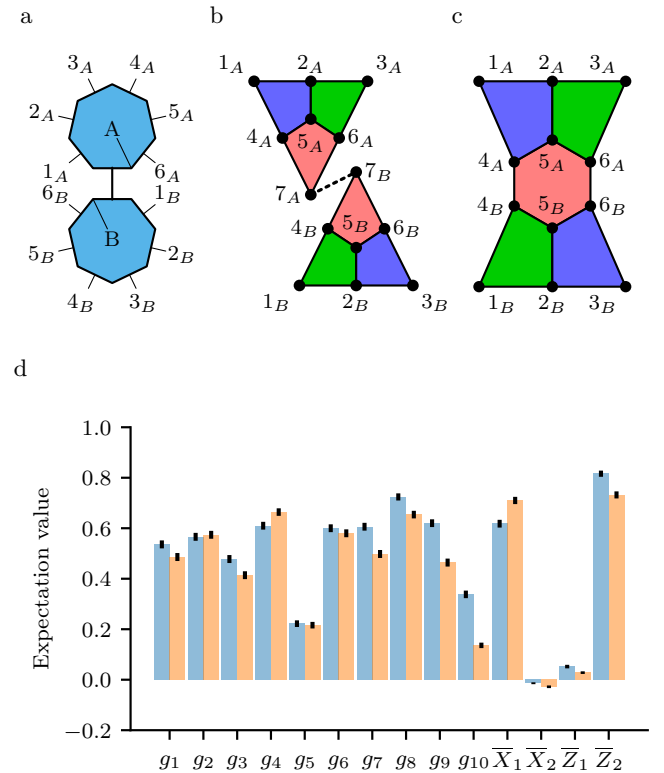
\begin{figure}
	\centering
	\begin{subfigure}[t]{0.24\linewidth}
		\captionsetup{justification=raggedright,
			singlelinecheck=false,
			position=top}
		\caption{{\small a}}
		\begin{tikzpicture}[line cap=round, line join=round, font=\fontsize{8pt}{9.6pt}\selectfont, scale=0.73]

  \def\R{0.8}      
  \def\gap{2.}     
  \def\tleg{1.38}   
  \def\tlabel{1.8} 
  \def\rotA{90}     
  \def\rotB{270}    

  \definecolor{mycolor}{RGB}{86,180,233}
  
  \coordinate (CA) at (0,0);
  \coordinate (CB) at (0,-\gap);

  \foreach \k in {0,...,6}{
    \coordinate (A\k) at ({\R*cos(\rotA-360*\k/7)},{\R*sin(\rotA-360*\k/7)});
  }
  \draw[fill=mycolor,draw=black,thick]
    (A0)--(A1)--(A2)--(A3)--(A4)--(A5)--(A6)--cycle;
  \node at (CA) {A};

  \foreach \k in {0,...,6}{
    \coordinate (B\k) at ({\R*cos(\rotB-360*\k/7)},{-\gap+\R*sin(\rotB-360*\k/7)});
  }
  \draw[fill=mycolor,draw=black,thick]
    (B0)--(B1)--(B2)--(B3)--(B4)--(B5)--(B6)--cycle;
  \node at (CB) {B};

  \draw[thick] ($(CA)+(0,-\R+0.09)$) -- ($(CB)+(0,\R-0.09)$);


  \foreach \k/\lab in {4/$1_A$,5/$2_A$,6/$3_A$,0/$4_A$,1/$5_A$,2/$6_A$}{
    \pgfmathtruncatemacro{\kp}{mod(\k+1,7)}
    \coordinate (Amid\k) at ($(A\k)!0.5!(A\kp)$);
    \draw (Amid\k) -- ($(CA)!\tleg!(Amid\k)$);
    \node at ($(CA)!\tlabel!(Amid\k)$) {\lab};
  }
  \draw (A3) -- ($(A3)!0.75!(CA)$);

  \foreach \k/\lab in {4/$1_B$,5/$2_B$,6/$3_B$,0/$4_B$,1/$5_B$,2/$6_B$}{
    \pgfmathtruncatemacro{\kp}{mod(\k+1,7)}
    \coordinate (Bmid\k) at ($(B\k)!0.5!(B\kp)$);
    \draw (Bmid\k) -- ($(CB)!\tleg!(Bmid\k)$);
    \node at ($(CB)!\tlabel!(Bmid\k)$) {\lab};
  }
  \draw (B3) -- ($(B3)!0.75!(CB)$);

\end{tikzpicture}
	\end{subfigure}
	\quad
    \begin{subfigure}[t]{0.32\linewidth}
		\captionsetup{justification=raggedright,
			singlelinecheck=false,
			position=top}
		\caption{{\small b}}        \definecolor{mytickcolor}{rgb}{0,0.45,0.70}

\begin{tikzpicture}[line cap=round, line join=round, font=\fontsize{8pt}{9.6pt}\selectfont, scale=0.3]

    \coordinate (Start) at (0,0);
    \coordinate (offset) at (2,-4);
    \tikzmath{real \L; \L = 5;}

    \coordinate (1A) at (0,{\L});
    \coordinate (3A) at ({\L},{\L});
    \coordinate (2A) at ($(1A)!0.5!(3A)$);
    \coordinate (7A) at ($(2A) - (0, {\L})$);

    \coordinate (4A) at ($(1A)!0.5!(7A)$);
    \coordinate (6A) at ($(3A)!0.5!(7A)$);
    \coordinate (5A) at ($ 1/3*(1A) + 1/3*(3A) + 1/3*(7A) $);

    \coordinate (1B) at (offset);
    \coordinate (3B) at ($(offset) + ({\L}, 0)$);
    \coordinate (2B) at ($(1B)!0.5!(3B)$);
    \coordinate (7B) at ($(2B)+(0, {\L})$);

    \coordinate (4B) at ($(1B)!0.5!(7B)$);
    \coordinate (6B) at ($(3B)!0.5!(7B)$);
    \coordinate (5B) at ($ 1/3*(1B) + 1/3*(3B) + 1/3*(7B) $);

  \fill[blue!60, draw=black, thick]  (1A) -- (2A) -- (5A) -- (4A) -- cycle;     
  \fill[green!80!black, draw=black, thick] (2A) -- (3A) -- (6A) -- (5A) -- cycle;     
  \fill[red!50, draw=black, thick]   (5A) -- (6A) -- (7A) -- (4A) -- cycle;     
  \fill[red!50, draw=black, thick]   (7B) -- (6B) -- (5B) -- (4B) -- cycle;     
  \fill[green!80!black, draw=black, thick] (4B) -- (5B) -- (2B) -- (1B) -- cycle;     
  \fill[blue!60, draw=black, thick]  (5B) -- (6B) -- (3B) -- (2B) -- cycle;     
  \draw[dotted, very thick] (7A) -- (7B);

  \foreach \p in {1A,2A,3A,4A,5A,6A,7A,7B,1B,2B,3B,4B,5B,6B}{
  \fill[black] (\p) circle (7.3pt);}

  \path (1A) node[above left]  {$1_A$};
  \path (2A) node[above]       {$2_A$};
  \path (3A) node[above right] {$3_A$};

  \path (4A) node[left]        {$4_A$};
  \path (5A) node[below=2.5] {$5_A$};
  \path (6A) node[right]       {$6_A$};

  \path (7A) node[left]        {$7_A$};
  \path (7B) node[right]       {$7_B$};

  \path (4B) node[left]        {$4_B$};
  \path (5B) node[above=2.5]       {$5_B$};
  \path (6B) node[right]       {$6_B$};

  \path (1B) node[below left]  {$1_B$};
  \path (2B) node[below]       {$2_B$};
  \path (3B) node[below right] {$3_B$};
\end{tikzpicture}
	\end{subfigure}
	\quad
    \begin{subfigure}[t]{0.32\linewidth}
		\captionsetup{justification=raggedright,
			singlelinecheck=false,
			position=top}
		\caption{{\small c}}
        \begin{tikzpicture}[x=1cm, y=1cm,,line cap=round, line join=round, font=\fontsize{8pt}{9.6pt}\selectfont, scale=0.97]
  \def\Rh{0.6} 
  \def\yTop{1.4}
  \def\yBot{-1.4}
  \def\xLeft{-1.}
  \def\xMid{0.0}
  \def\xRight{1.}

  \coordinate (H0) at (0,\Rh);                                 
  \coordinate (H1) at ({-\Rh*sqrt(3)/2},{\Rh/2});              
  \coordinate (H2) at ({-\Rh*sqrt(3)/2},{-\Rh/2});             
  \coordinate (H3) at (0,-\Rh);                                
  \coordinate (H4) at ({\Rh*sqrt(3)/2},{-\Rh/2});              
  \coordinate (H5) at ({\Rh*sqrt(3)/2},{\Rh/2});               

  \coordinate (TL) at (\xLeft,\yTop);
  \coordinate (TM) at (\xMid,\yTop);
  \coordinate (TR) at (\xRight,\yTop);
  \coordinate (BL) at (\xLeft,\yBot);
  \coordinate (BM) at (\xMid,\yBot);
  \coordinate (BR) at (\xRight,\yBot);

  \draw[fill=red!50, draw=black, thick] (H0)--(H1)--(H2)--(H3)--(H4)--(H5)--cycle;

  \draw[fill=blue!60, draw=black, thick] (TL)--(TM)--(H0)--(H1)--cycle; 
  \draw[fill=green!80!black, draw=black, thick] (TM)--(TR)--(H5)--(H0)--cycle; 

  \draw[fill=green!80!black, draw=black, thick] (BL)--(BM)--(H3)--(H2)--cycle; 
  \draw[fill=blue!60, draw=black, thick] (BM)--(BR)--(H4)--(H3)--cycle;  

  \foreach \P in {TL,TM,TR,BL,BM,BR,H1,H0,H5,H2,H3,H4}{
    \fill[black] (\P) circle (2.2pt);
  }

  \node[anchor=south west] at (TL) {$1_A$};
  \node[anchor=south]      at (TM) {$2_A$};
  \node[anchor=south east] at (TR) {$3_A$};

  \node[anchor=north west] at (BL) {$1_B$};
  \node[anchor=north]      at (BM) {$2_B$};
  \node[anchor=north east] at (BR) {$3_B$};

  \node[anchor=east]  at (H1) {$4_A$};
  \node[anchor=north, yshift=-2] at (H0) {$5_A$};
  \node[anchor=west]  at (H5) {$6_A$};

  \node[anchor=east]  at (H2) {$4_B$};
  \node[anchor=south, yshift=2] at (H3) {$5_B$};
  \node[anchor=west]  at (H4) {$6_B$};

\end{tikzpicture}
	\end{subfigure}
	\quad
	\begin{subfigure}[t]{\linewidth}
		\captionsetup{justification=raggedright,
			singlelinecheck=false,
			position=top}
		\caption{{\small d}}
		\includegraphics[width=\linewidth]{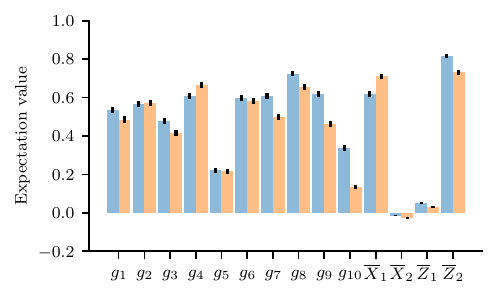}
	\end{subfigure}
	 \caption{\textbf{Holographic heptagon code.}
    \textbf{(a)} Small instance of the heptagon code we implement in the experiment. \textbf{(b)} Two $[\![7, 1, 3]\!]$ color codes are contracted along qubits 7A and 7B (dashed line) to obtain the \textbf{(c)} $[\![12, 2, 3]\!]$ Heptagon code. 
    Note that the color code representation in (c) indicates the stabilizer operators, 
    but does not preserve the inherent bulk-boundary geometry from which it originates, that is visible in (a).
    \textbf{(d)}~The simulated (orange) and measured (blue) expectation values of the stabilizer generators and logical operators for the logical $\left|\overline{\text{+0}}\right>_\text{hep}-$state. The error bars are computed from quantum projection noise and represent one standard deviation of statistical uncertainty. The logical operators are post-selected depending on the sign of the weight-6 stabilizers in Eq.~\eqref{eq:Hadamard_correction_heptagon} mapped onto two ancilla qubits.} 
    \label{fig:heptagon}
\end{figure}

\subsection{Heptagon Code}\label{subsection_pentagon_encoding}
While we were investigating holographic features like partial decoding using the pentagon code, we focus here more on the aspect of quantum error detection and subsequently correction, harnessing the efficient mid-circuit capabilities of the heptagon code. In the following, we investigate the initialization of the code as well as implementing single qubit and two-qubit gates.

{\bf Logical operations. ---}
We focus on a small instance of the holographic heptagon code, shown in Fig.~\ref{fig:heptagon}, encoding two logical qubits, labeled A and B, on 12 physical boundary qubits.  Here, we aim to implement the logical zero state $\ket{\overline{00}}_\text{hep}$ of the holographic heptagon code. The success of the implementation is measured by evaluating the expectation values of the stabilizer and the logical operators, which are given in Supp. Mat.~\ref{SS:Stabilizer_generators_logical_operators_heptagon}. We find an average expectation value of all stabilizer generators and logical operators of 0.71(10) [sim: 0.62(14)]. All $X/Z$-type stabilizer generators and logical operators are calculated simultaneously for each shot in their respective measurement basis. A total of $2\cdot10^4$ shots were taken per measurement basis.

The two logical qubits can be entangled by applying a transversal MS gate or by applying a Hadamard on the control qubit and sub-sequentially a transversal CNOT gate. The average expectation value of the stabilizer generators and logical operators after applying the MS gate is 0.48(15) [sim: 0.37(9)]. The measurements and the simulations are displayed in Fig.\ref{fig:Heptagon_init}(a-c).

{\bf Logical Hadamard gate. ---}
The implementation of the logical Hadamard gate is not straightforward. Applying a logical Hadamard gate on either one of the logical qubits will switch the weight 4 stabilizer generators from the X-basis to the Z-basis and vice versa for the corresponding logical qubit, while the stabilizer generators for the other qubit stay the same. However, this is not true for the weight 6 stabilizer generators. Considering the product of the weight 6 stabilizer generators before
\begin{align}
    g_5g_{10} =& (-i)Y_{4_A}(-i)Y_{5_A}(-i)Y_{6_A}iY_{4_B}iY_{5_B}iY_{6_B} \notag \\
    =& Y_{4_A}Y_{5_A}Y_{6_A}Y_{4_B}Y_{5_B}Y_{6_B}
\end{align}
and after the application of the logical Hadamard gate
\begin{align}
    g'_5g'_{10} =& (-i)Y_{4_A}(-i)Y_{5_A}(-i)Y_{6_A}(-i)Y_{4_B}(-i)Y_{5_B}(-i)Y_{6_B} \notag\\
    =& -Y_{4_A}Y_{5_A}Y_{6_A}Y_{4_B}Y_{5_B}Y_{6_B}
\end{align}
shows a change of sign after the application of the logical Hadamard gate. Hence, if we measure $g_5$ and $g_{10}$ on the state $H_A\ket{\psi}$, they will always have opposite signs. The induced error can be fixed by a round of QEC. Depending on which stabilizer generator has the negative sign, one of the following corrections $C$ can be applied
\begin{align}\label{eq:Hadamard_correction_heptagon}
    g_5 &= +1 g_{10} = -1 \rightarrow C = X_{4_B}X_{5_B}X_{6_B}, \notag \\
    g_5 &= -1 g_{10} = +1 \rightarrow C = Z_{4_B}Z_{5_B}Z_{6_B}.
\end{align}
We choose to map the sign of the two stabilizer generators $g_5$ and $g_{10}$ to two ancilla qubits in order to experimentally demonstrate the Hadamard gate. We initialize the logical $\ket{\overline{00}}_\text{hep}$ state on the 6 leftmost and rightmost ions of the 16 ions available in the linear Paul trap and apply a logical Hadamard on logical qubit A. 
The induced error on qubit B is detected by performing a parity measurement on the ancilla qubits, which are located next to the ions used to initialize the logical qubits, and consequentially corrected via Pauli-frame update~\cite{gottesman1997stabilizercodesquantumerror, Roffe03072019} in post-processing. The average expectation value of the stabilizer generators and the logical operators $\overline{X}_A$ and $\overline{Z}_B$ is $0.56(15)$ [sim: $0.5(2)$]. The expectation values of all stabilizers and logical operators after performing the Pauli-frame update are shown in Fig.~\ref{fig:heptagon}(c).

\section{Outlook}

Recent experiments have largely focused  on the class of
{\em topological} codes, 
which can correct for all topologically trivial errors.
In contrast, 
holographic codes are of interest due to their {\em geometric} properties,
where partial recovery of the encoded quantum information is possible even if large parts of the system are inaccessible or affected by errors.
This scenario, for example, is relevant when dealing with burst errors, where multiple consecutive qubits are affected by errors~\cite{fan2025characterizingbursterrorcorrection,doi:10.1143/JPSJ.69.3540}.
Indeed, it has been shown that holographic codes can span a wide range of rates, in particular achieving finite rates and competitive error thresholds~\cite{
steinberg2024farperfectquantumerror}.
However, current constructions require syndrome measurements of high weights, and it is an open question whether it is possible to create holographic codes 
that require low- or constant-weight stabilizers only that are of practical use. 

Our work demonstrates the realization of a holographic quantum code with bulk reconstruction. 
In a holographic code, the ability to perform partial recovery is directly linked to geometry and entanglement.
It is a natural question what kind of geometries of many-body entangled states and associated cuts might support similar quantum error correction and recovery capabilities. While initial steps in this direction suggest monogamy of entanglement as a key constraint~\cite{Hayden_2013, Umemoto_2018, Bao_2015}, much remains to be understood. Our implementation thus provides a key step towards understanding and simulating holography and bulk-boundary duality, and paves the way for the use of holographic codes in state-of-the-art quantum computers.

\vspace{0.5cm}
\section*{Acknowledgements}
We thank 
Valentin Kasper, for fruitful discussions. GM and FH were supported by the Foundation for Polish Science through TEAM-NET (POIR.04.04.00-00-17C1/18-00). 
GM was also supported by NCN grant no. 2024/53/B/ST2/04103.
FH was also funded in whole or in part by the National Science Centre, Poland 2024/54/E/ST2/00451
and by the Polish National Agency for Academic Exchange under the Strategic Partnership Programme grant BNI/PST/2023/1/00013/U/00001.
For the purpose of Open Access,
the author has applied a CC-BY public copyright licence to any Author Accepted Manuscript (AAM) version arising from this submission.
This research was funded by the European Research Council (ERC, QUDITS, 101039522). Views and opinions expressed are however those of the author(s) only and do not necessarily reflect those of the European Union or the European Research Council Executive Agency. Neither the European Union nor the granting authority can be held responsible for them. We also acknowledge support by the Austrian Science Fund (FWF) through the SFB BeyondC (FWF Project No. F7109) and the EU-QUANTERA project TNiSQ (N-6001), and by the IQI GmbH, and by the European Union’s Horizon Europe research and innovation programme under grant agreement No 101114305 (“MILLENION-SGA1” EU Project). 
RH and TMS acknowledge support from the Australian Research Council Centres of Excellence for Engineered Quantum Systems (Grants CE110001013 and  CE170100009).

Note added.– While completing this work, we became aware of Ref.~\cite{biswas2026observationgravitylikesignaturesholographic} by Biswas, D. et al., studying the HaPPY code in the context of quantum gravity on a trapped-ion quantum computer.

\section*{Conflicts of interest}
T.M. is CEO at Alpine Quantum Technologies GmbH, a commercially oriented quantum computing company.

\section*{Author contributions}
A.S. led and conducted the experimental aspects of the work and contributed to the experimental parts in writing. 
G. A. M. developed the theoretical framework of the pentagon code, and contributed to its writing.
R.F., I.P. contributed to building and maintaining the experimental apparatus. 
M. M. contributed to implementing and debugging the experiments. 
R. H., G. B., and T. S. developed the theoretical framework of the heptagon code, including the theoretical formulation of the Hadamard gate.
F. H., M. R., R.B., and T.M. conceptualized the research and implementation. F.H. and M.R. coordinated the research activities across the institutions.
All the authors contributed to the analysis and interpretation of results, and the writing of the manuscript.

\section*{Data availability statement}
Data and analysis scripts can be made available on request. 
Please contact alex.steiner@uibk.ac.at.

\bibliographystyle{apsrev4_2_mod.bst}
\bibliography{bibliography}

\newpage
\clearpage

\appendix

\setcounter{figure}{0}
\setcounter{equation}{0}
\setcounter{table}{0}
\setcounter{section}{0}

\makeatletter
\renewcommand{\theequation}{\thesection\@arabic\c@equation}
\renewcommand{\thefigure}{\thesection\@arabic\c@figure}
\renewcommand{\thetable}{\thesection\@arabic\c@table}
\renewcommand{\thesubsection}{\thesection\@arabic\c@subsection}
\renewcommand{\p@subsection}{}
\makeatother
\footnotesize

\renewcommand\appendixname{Methods}
\renewcommand\appendixpagename{Methods}

\renewcommand{\theequation}{M\arabic{equation}}
\renewcommand{\thefigure}{M\arabic{figure}}
\renewcommand{\thetable}{M\Roman{table}}
\renewcommand{\thesection}{M\Roman{section}}
\renewcommand{\thesubsection}{M\Roman{section} \Roman{subsection}}

\renewcommand{\theHequation}{M\arabic{equation}}
\renewcommand{\theHfigure}{M\arabic{figure}}
\renewcommand{\theHtable}{M\Roman{table}}
\renewcommand{\theHsection}{M\Roman{section}}
\renewcommand{\theHsubsection}{M\Roman{section} \Roman{subsection}}

\setcounter{equation}{0}
\setcounter{figure}{0}
\setcounter{table}{0}
\setcounter{section}{0}

\renewcommand\appendixname{Supplementary}
\renewcommand\appendixpagename{Supplementary}

\renewcommand{\theequation}{S\arabic{equation}}
\renewcommand{\thefigure}{S\arabic{figure}}
\renewcommand{\thetable}{S\Roman{table}}
\renewcommand{\thesection}{S\Roman{section}}
\renewcommand{\thesubsection}{S\Roman{section} \Roman{subsection}}

\renewcommand{\theHequation}{S\arabic{equation}}
\renewcommand{\theHfigure}{S\arabic{figure}}
\renewcommand{\theHtable}{S\Roman{table}}
\renewcommand{\theHsection}{S\Roman{section}}
\renewcommand{\theHsubsection}{S\Roman{section} \Roman{subsection}}

\setcounter{equation}{0}
\setcounter{figure}{0}
\setcounter{table}{0}
\setcounter{section}{0}
\normalsize
\makeatletter
\renewcommand*\env@matrix[1][*\c@MaxMatrixCols c]{%
	\hskip -\arraycolsep
	\let\@ifnextchar\new@ifnextchar
	\array{#1}}
\makeatother

\newcommand{\BB}{{\partial \mathcal{B}}}
\newcommand{\MM}{\mathcal{B}}

\newcommand\overlinelong[1]{\accentset{\rule{0.9em}{.8pt}}{#1}}
\newcommand{\CZfbar}[0]{{\operatorname{\overlinelong{CZ}}}}

\newcommand{\pH}[0]{{H}}
\newcommand{\pS}[0]{{S}}
\newcommand{\Hbar}[0]{{\overline{H}}}
\newcommand{\Sbar}[0]{{\overline{S}}}
\newcommand{\Ybar}[0]{{\overline{Y}}}

\newcommand{\Un}[0]{U}
\newcommand{\Vop}[0]{V}
\newcommand{\Wop}[0]{W}
\newcommand{\CZ}[0]{\operatorname{CZ}}
\newcommand{\CX}[0]{\operatorname{CX}}
\newcommand{\CNOT}[0]{\operatorname{CNOT}}
\newcommand{\SWAP}[0]{\operatorname{SWAP}}
\newcommand{\pX}[0]{X}
\newcommand{\pZ}[0]{Z}
\newcommand{\pY}[0]{Y}
\newcommand{\Had}[0]{H}
\newcommand{\CZbar}[0]{\operatorname{\overline{CZ}}}
\newcommand{\CXbar}[0]{\operatorname{\overline{CX}}}
\newcommand{\Xbar}[0]{\overline{X}}
\newcommand{\Zbar}[0]{\overline{Z}}
\newcommand{\bmR}[1]{\bm{\mathrm{#1}}}

\newcommand{\N}{\mathds{N}}
\newcommand{\F}{\mathds{F}}
\newcommand{\Q}{\ma\thds{Q}}

\newcommand{\EE}{\mathcal{E}}
\newcommand{\HH}{\mathcal{H}}
\newcommand{\KK}{\mathcal{K}}
\newcommand{\XX}{\mathcal{X}}
\newcommand{\YY}{\mathcal{Y}}
\newcommand{\UU}{\mathcal{U}}

\newcommand{\PP}{\mathcal{P}}
\newcommand{\DD}{\mathcal{D}}
\newcommand{\oper}[1]{\operatorname{#1}}

\renewcommand{\a}{\alpha}
\renewcommand{\b}{\beta}
\renewcommand{\SS}{\mathcal{S}}
\newcommand{\II}{\mathcal{I}}
\renewcommand{\AA}{\mathcal{A}}
\newcommand{\nh}{\lfloor \tfrac{n}{2} \rfloor}
\newcommand{\Ml}{M^{(\ell)}}
\newcommand{\Al}{A^{(\ell)}}
\newcommand{\Bl}{B^{(\ell)}}

\newcommand{\RR}{\mathcal{R}}
\newcommand{\QQ}{\mathcal{M}}
\newcommand{\nn}{\nonumber}

\newcommand{\nothing}[1]{#1}

\makeatletter

\@removefromreset{equation}{section}
\@removefromreset{figure}{section}
\@removefromreset{table}{section}

\renewcommand{\theequation}{S\arabic{equation}}
\renewcommand{\thefigure}{S\arabic{figure}}
\renewcommand{\thetable}{S\arabic{table}}

\makeatother


\section{State initialization, logical operators, and stabilizer generators of the pentagon code}\label{ss:Stabilizer_generators_logical_operators_HAPPY}

Following Ref.~\cite{Angl_s_Munn__2024}, we now outline how to prepare the logical zero state $\ket{G^{}_{0000}}$ of the holographic pentagon code experimentally. First, we initialize 12 physical qubits in the $\ket{+}$ state. Then, we apply 28 $\text{CZ}$-gates on the physical qubits in order to create the logical ground state:
\begin{equation}
    \ket{G^{}_{0000}} = \prod_{(u, v)\in\text{E}_{\partial \mathcal{B}}}\text{CZ}_{uv}\ket{+}^{\otimes12}.
\end{equation}
 Here, $\text{E}_{\partial B}$ is the set containing the 28 $\text{CZ}$-gates applied on the respective physical qubits on the boundary (denoted $\partial B$) of the holographic code. A graphical demonstration of $\text{E}_{\partial B}$ is given in Fig.~\ref{fig:MinimalGraph}(b).
\begin{align*}
    \text{E}_{\partial B} =&\{(2, 1), (2, 4), (2, 5), (2, 6), (3, 4), (3, 5), (3, 6),\\&(5, 4), (5, 7), (5, 8), (5, 9), (6, 7), (6, 8), (6, 9),\\&(8, 7), (8, 10), (8, 11), (8, 12), (9, 10), (9, 11),\\&(9, 12), (11, 10), (11, 1), (11, 2), (11, 3),\\&(12, 1), (12, 2), (12, 3)\}.
\end{align*}

The logical operators of the pentagon code are defined as follows,

\begin{tabbing}
    \hspace{1.5cm}\=$\overline{Z}_1=Z_{1}X_{2}X_{3}$, \hspace{0.7cm}\=$\overline{Z}_2=Z_{4}X_{5}X_{6}$,\\\>$\overline{Z}_3=Z_{7}X_{8}X_{9}$, \>$\overline{Z}_4=Z_{10}X_{11}X_{12}$\\\\
    \>$\overline{X}_1=-Y_{1}Z_{2}Y_{3}$, \>$\overline{X}_2=-Y_{4}Z_{5}Y_{6}$,\\ \>$\overline{X}_3=-Y_{7}Z_{8}Y_{9}$, \>$\overline{X}_4=-Y_{10}Z_{11}Y_{12}$\\\\
    \>$\overline{Y}_1=-X_{1}Y_{2}Z_{3}$, \>$\overline{Y}_2=-X_{4}Y_{5}Z_{6}$,\\ \>$\overline{Y}_3=-X_{7}Y_{8}Z_{9}$, \>$\overline{Y}_4=-X_{10}Y_{11}Z_{12}$
\end{tabbing}

The stabilizer generators of the subspace are
\begin{tabbing}
$g_1 = X_1 Z_2 X_3 Z_4 Z_5 Z_6$, \hspace{0.7cm} \=$g_2=X_4 Z_5 X_6 Z_7 Z_8 Z_9$,\\ $g_3 = X_7 Z_8 X_9 Z_{10} Z_{11} Z_{12}$, \>$g_4 = Z_{1} Z_{2} Z_{3} X_{10} Z_{11} X_{12}$\\$g_5 = Y_1 X_2 Z_3 X_4 X_5 Y_6$, \>$g_6=Y_4 X_5 Z_6 X_7 X_8 Y_9$,\\$g_7 = Y_7 X_8 Z_9 X_{10} X_{11} Y_{12}$, \>$g_8 = X_{1} X_{2} Y_{3} Y_{10} X_{11} Z_{12}$\\
\end{tabbing}

\section{Initialization of the holographic heptagon code}\label{SS:Init_heptagon}
The holographic heptagon code is initialized on the edges of the register with 16 qubits $q$ as following: $1_A\rightarrow q_1$, $2_A\rightarrow q_2$, $3_A\rightarrow q_3$, $4_A\rightarrow q_4$, $5_A\rightarrow q_5$, and $6_A\rightarrow q_0$ for logical qubit A and $1_B\rightarrow q_{11}$, $2_B\rightarrow q_{12}$, $3_B\rightarrow q_{13}$, $4_B\rightarrow q_{14}$, $5_B\rightarrow q_{15}$, and $6_B\rightarrow q_{10}$ for logical qubit B. The ancilla qubits required for implementing the Hadamard gate are $A_1=q_6$ and $A_2=q_9$.
A total of 5 Hadamard gates and 17 CNOT gates are required to initialize the logical zero state $\ket{\overline{00}}_\text{hep}$. The process is described by the following equation:
\begin{equation}
    \ket{\overline{00}}_\text{hep}= \left(\prod_{(u, v) \in G}\text{CNOT}_{uv}\right)\prod_{k\in K}H_k\ket{0}^{\otimes12},
\end{equation}
where $K$ is a set containing the physical qubits on which the Hadamard gate is applied and $G$ is a set containing information on which physical qubits the CNOT gate is applied: 
\begin{align}
    K =& \left\{0, 1, 3, 7, 9\right\},\notag\\
    G=&\{(0, 4),(7, 8),(1, 2),(0,5),(1,4),(0,6),(1,5),\\
    &(9,6),(0,10),(0,11),(3,0),(7,10),(3,2),\notag\\
    &(9,8),(3,4),(7,11),(9,10)\}.\notag
\end{align}

\section{Simulation Methods}\label{SS:Simulation_methods}
We model circuit errors as depolarizing errors, which reproduces well the experimentally observed infidelities despite its conceptual simplicity, which does not take the microscopic physical processes underlying noisy gates and operations in the ion trap into account explicitly. Noise is applied in simulations by randomly placing Pauli errors $E$ according to the experimental physical error rates after every single-qubit operation as well as two-qubit gates. The errors can be 
\begin{align}
    E_1 &\in \{\sigma_k, \forall k \in \{1, 2, 3\}  \}\,,\nn \\
    E_2 &\in \{\sigma_k \otimes \sigma_l, \forall k, l \in \{0, 1, 2, 3\}  \} \backslash \{I\otimes I\} \,,
\end{align}
where $\sigma_k = \{I, X, Y, Z\}$ with $k=0, 1, 2, 3$ are the Pauli matrices. The error channels for our depolarizing noise
model read
\begin{align}
    \mathcal{E}_1(\rho)&=(1-p_1)\rho + \frac{p_1}{3}\left(X\rho X+Y\rho Y+Z\rho Z\right)\,,\nn\\
    \mathcal{E}_2(\rho)&=(1-p_2)\rho + \frac{p_2}{15}\left[\sum_{i, j=0}^{3}(\sigma_i\rho\sigma_j-\rho)\right]\,,
\end{align}
 so that any single-qubit error is applied uniformly to the ideal operation with equal probability $p_1/3$ and the single-qubit operation is executed ideally with probability $1 - p_1$; two-qubit errors are applied uniformly after the ideal two-qubit gates with equal probability $p_2/15$ and any two-qubit gate is executed ideally with probability $1-p_2$. In all simulations we used physical error rates of
\begin{align}\label{Eq:gate_infidelity}
    p_1&=0.005\\
    p_2&=0.025\notag
\end{align}
for the corresponding operations \cite{pogorelov_compact_ion_2021}.

\section{Additional logical operators of the pentagon code}\label{ss:Additional_logical_operators_HaPPY}
In this section, we obtain additional logical operators 
(combinations of the Hadamard $H$, Phase $S$ and CZ gates) 
from the graph code proposed in~\cite{Angl_s_Munn__2024}.
We recall that such code corresponds to a small instance of the holographic code on $12$ qubits, labeled as shown in Fig.~\ref{fig:unitary_rec}(c). 
The obtained operators are used to prepare the following two logical single-qubit states, $\ket{G^{}_{+000}}$ and $\ket{G^{}_{+_y 000}}$.
Furthermore, we show how to prepare a logical entangled two-qubits state,
${1/\sqrt{2}}\left(\ket{G^{}_{+0++}} + \ket{G^{}_{-1++}}\right)$.
The experimental realization of these states is done in Section~\ref{subsection_happy_gates}.
The required logical gates are
the logical Hadamard gate $\Hbar_A$ and the logical controlled $Z$ gate $\CZbar_{\operatorname{AB}}$ that are defined to act on the code subspace $\CC \subseteq \HH_{\BB}$ in the same way as $H_A$ and $\CZ_{AB}$ act on $\HH_{\MM}$. 
Due to the symmetry of the graph code, we also obtain two logical gates that act locally on the boundary qubits and rotate the logical basis.

\bigskip
\noindent {\bf Invariant local gates.} 
Here, we derive two logical gates that rotate the logical basis.
First we describe the action of three one-qubit gates.

The Hadamard and Phase gates are
\begin{align}
    H=  \frac{1}{\sqrt{2}}\begin{pmatrix}
        1 & -1 \\
        -1 & 1
    \end{pmatrix} \,, \quad 
    S=	\begin{pmatrix}
        1 & 0 \\
        0 & -i
    \end{pmatrix} \,.
\end{align}
These act by conjugation 
as $HXH^\dagger = Z$, $HZH^\dagger = X$, $HYH^{\dagger} = -Y$; and $SXS^\dagger= Y$,  $SZS^\dagger = Z$ and  $SYS^\dagger = -X$.

Combining $H$ and $S$ gates, they act by conjugation as,

\begin{equation}
    \begin{matrix}
        HS: & \pX &\rightarrow& - \pY\,,  &&&
        SH: & \pX &\rightarrow& \pZ\,, \\
        &\pY &\rightarrow& -\pZ\,,&&&
        &\pY &\rightarrow&  \pX\,, \\
        &\pZ &\rightarrow& \phantom{-}\pX \,,&&&
        &\pZ &\rightarrow  &\pY\,.
    \end{matrix}
\end{equation}
Taking that into account, we define the following gate $R$ acting on the boundary,
\begin{align}\label{eq:R_rot}
    \begin{matrix}[cc ccc ccc]
        &&\text{\footnotesize1}&\text{\footnotesize2}&\text{\footnotesize3}&\;
        \text{\footnotesize4}&\text{\footnotesize5}&\text{\footnotesize6}\\
        R&=&
        \pS\pH\pY &\pS\pH\pZ& \pH\pS&\,
        \pS\pH\pY &\pS\pH\pZ& \pH\pS \\ 
        &&\text{\footnotesize7}&\text{\footnotesize8}&\text{\footnotesize9}&\; \text{\footnotesize10}&\!\text{\footnotesize11}&\!\!\!\!\text{\footnotesize12}\\
        &&\pS\pH\pY &\pS\pH\pZ& \pH\pS&\,
        \pS\pH\pY &\pS\pH\pZ& \pH\pS \,. 
    \end{matrix}
\end{align}
One can check that $R$ leaves the code subspace invariant. 
The generators of the code subspace $g_1,\dots, g_8$ can be found in Section~\ref{ss:Stabilizer_generators_logical_operators_HAPPY}.
This operator rotates by conjugate action 
the logical gates in the following way:

\begin{equation}
    \begin{aligned}
        {\Xbar_A} \xrightarrow{{R}} - {\Ybar_A} \xrightarrow{{R}} {\Zbar_A} \xrightarrow{{R}} {\Xbar_A} \\
        {\Xbar_B} \xrightarrow{{R}} -{\Ybar_B} \xrightarrow{{R}} {\Zbar_B} \xrightarrow{{R}} {\Xbar_B} \\
        {\Xbar_C} \xrightarrow{{R}} - {\Ybar_C} \xrightarrow{{R}} {\Zbar_C} \xrightarrow{{R}} {\Xbar_C} \\
        {\Xbar_D} \xrightarrow{{R}} - {\Ybar_D} \xrightarrow{{R}} {\Zbar_D} \xrightarrow{\operatorname{R}} {\Xbar_D}
    \end{aligned}
\end{equation}
Thus, applying $\operatorname{R}$ to the boundary is a equivalent to applying $\pH\pS \otimes \pH\pS  \otimes \pH\pS \otimes \pH\pS $ to the bulk qubits.  
Similarly, define the invariant gate $\operatorname{L}$,

\begin{equation}\label{eq:L_rot}
    \begin{matrix}[cc ccc ccc]
        &&\text{\footnotesize1}&\text{\footnotesize2}&\text{\footnotesize3}&\;
        \text{\footnotesize4}&\text{\footnotesize5}&\text{\footnotesize6}\\
        \operatorname{L}&=&
        \pH\pS\pX &\pH\pS\pY & \pS\pH&\,
        \pH\pS\pX &\pH\pS\pY & \pS\pH \\
        &&\text{\footnotesize7}&\text{\footnotesize8}&\text{\footnotesize9}&\;\text{\footnotesize10}&\!\text{\footnotesize11}&\!\!\!\!\text{\footnotesize12} \\
        &&\pH\pS\pX &\pH\pS\pY & \pS\pH&\,
        \pH\pS\pX &\pH\pS\pY & \pS\pH \,. \\
    \end{matrix}
\end{equation}
Again, this operator rotates the logical gates,
\begin{equation}
    \begin{aligned}
        {\Xbar_A} \xleftarrow{{L}} {\Ybar_A} \xleftarrow{{L}} {\Zbar_A} \xleftarrow{\operatorname{L}} {\Xbar_A} \\
        {\Xbar_B} \xleftarrow{{L}} {\Ybar_B} \xleftarrow{{L}} {\Zbar_B} \xleftarrow{\operatorname{L}} {\Xbar_B} \\
        {\Xbar_C} \xleftarrow{{L}} {\Ybar_C} \xleftarrow{{L}} {\Zbar_C} \xleftarrow{{L}} {\Xbar_C} \\
        {\Xbar_D} \xleftarrow{{L}} {\Ybar_D} \xleftarrow{{L}} {\Zbar_D} \xleftarrow{{L}} {\Xbar_D}
    \end{aligned}
\end{equation}

Thus, applying ${L}$ to the boundary is equivalent to applying $\pS\pH \otimes \pS\pH \otimes \pS\pH  \otimes \pS\pH$ in the bulk. 

The gates ${L}$ and ${R}$ seem to be the only tensor-product gates that keep the code subspace invariant and, at the same time, change the logical gates. 
We propose to use these gates as an "effective" Hadamard gates in the creation of the logical zero state, 
since $L$ and $R$ only require local gates: applying $R$ to the logical zero state $\ket{G^{}_{0000}}$ yields to

\begin{equation}
    \ket{G^{}_{++++}}=\operatorname{R}\ket{G^{}_{0000}}\,.
\end{equation}
We will use these gates to obtain $\ket{G^{}_{+000}}$, $\ket{G^{}_{+_y 000}}$ and ${1/\sqrt{2}}\left(\ket{G^{}_{+0++}} + \ket{G^{}_{-1++}}\right)$ involving less two-qubit gates.

\bigskip
\noindent {\bf Logical Hadamard and CZ gates.}
A Hadamard gate can be expressed as $\pH=1/\sqrt{2}(\pX+\pY)$.
Thus, the logical Hadamard $\overline{\pH}$ gate can also be written in a similar way using the logical $\Xbar$ and $\Zbar$ gates.
We explicitly write the logical Hadamard acting on qubit A, but the procedure is equivalent for the rest of logical qubits:
\begin{align}\label{eq:log_H}
    \Hbar_A&=\frac{1}{\sqrt{2}}({\Xbar^{\oper{opt}}_A}+{\Zbar^{\oper{opt}}_A}) \nn \\
    &= \frac{1}{\sqrt{2}} (\pZ_1\pX_2\pX_3 - \pY_1 \pZ_2 \pY_3) \nn \\ 
    &=\pY_1\pH_3\CX_{23}\CX_{21}\pH_2\CX_{21}\CX_{23}\pY_3\pH_3
\end{align}
We can then obtain $\ket{G^{}_{+000}}=\Hbar_{\operatorname{A}}\ket{G^{}_{0000}}$.
Combining the rotational gates $R$ and $L$ together with $\Hbar_{\operatorname{A}}$, we also derive 
\begin{equation}\label{eq:logical_y} \ket{G^{}_{+_y000}}={L}\Hbar_{\operatorname{A}}{R} \ket{G^{}_{0000}}\,.
\end{equation}
This can be seen by the following scheme,
\begin{align}\label{eq:y_log}
    &\ket{G^{}_{0000}} \, \xrightarrow{R} \,\ket{G^{}_{++++}}\, \nn \\&\xrightarrow{\Hbar_{\operatorname{A}}} \,\ket{G^{}_{0+++}} \, \xrightarrow{L} \,  \ket{G^{}_{+_y000}} \,.
\end{align}

A $\CZ$ gate acting on two qubits labeled $i$ and $j$ can be written as $\CZ_{ij}=(1/2)(\one +\pZ_i+ \pZ_j- \pZ_i\pZ_j)$. Therefore, the logical $\CZbar$ can also be expressed as:
\begin{align}
    \CZbar_{\text{AB}} &=\frac{1}{2}(\one + \pZ_1\pX_2\pX_3+\pZ_4\pX_5\pX_6-\pZ_1\pX_2\pX_3\pZ_4\pX_5\pX_6)  \nn\\
    &=\pH_1\pH_4\frac{1}{2}(\one + \pX_1\pX_2\pX_3+\pX_4\pX_5\pX_6 \nn \\ 
    & \quad\, -\pX_1\pX_2\pX_3\pX_4\pX_5\pX_6)  \pH_1\pH_4 \nn\\
    &=(\pH_1\pH_4 \CX_{12}\CX_{13}\CX_{45}\CX_{46}\pH_1\pH_4 )\CZ_{14} \nn\\ &\quad\,\cdot(\pH_1\pH_4 \CX_{12}\CX_{13}\CX_{45}\CX_{46}  \pH_1\pH_4)\,.
\end{align}
This gate allows us to obtain an entangled logical state by applying gates on the boundary:
\begin{equation}\label{eq:ent_gate}
    \CZbar_{AB}\operatorname{R} \ket{G^{}_{0000}}=\frac{1}{\sqrt{2}}(\ket{G^{}_{+0++}}+\ket{G^{}_{-1++}})\,.
\end{equation}

\section{Logical operator and stabilizer values of single-qubit and entangling gates in the pentagon code}\label{SS:Eigenvalues_gates_HAPPY}
We show the measured logical operators and stabilizer generators for logical qubit $A$ measured in the Pauli bases $\{X, Y, Z\}$ and for a $\overline{\text{CZ}}$ gate applied on  logical qubit A and B on the state $\ket{\text{++++}}$ in Fig.~\ref{fig:SS_eigenvalues_happy_init}. Measurements are illustrated in blue bars while numerical simulation is given in orange. The measured expectation values of the measured states agree with the numerical simulation indicating the successful implementation of the corresponding gates. Here we show the manipulation of logical qubit $A$, but all gates can be applied on qubits $B-D$. 
\begin{figure}[!ht]
    \centering
    \includegraphics[width=\linewidth]{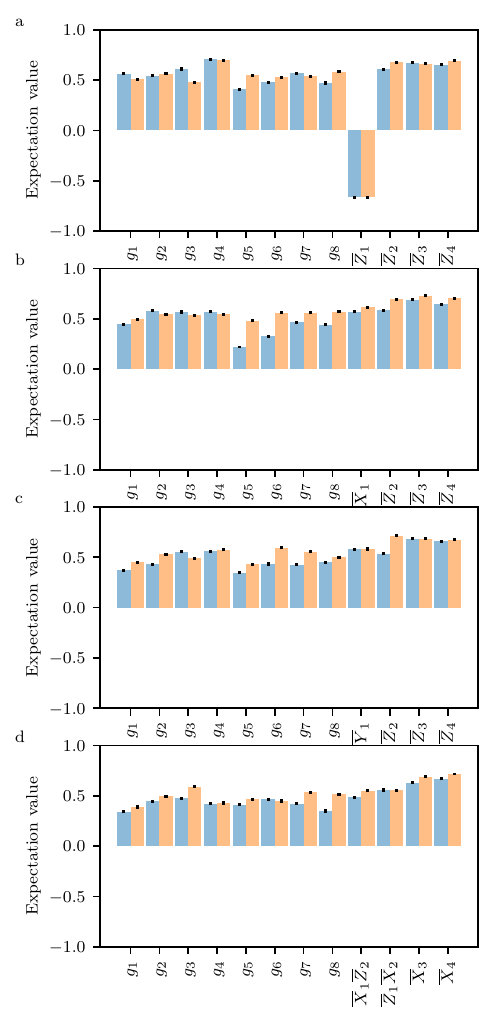}
    \caption{\textbf{State preparation.} \textbf{(a-c)} Expectation values of the stabilizer generators and logical operators initializing the logical qubit $A$ of $\ket{G^{}_{0000}}$ in the Pauli bases $\{X, Y, -Z\}$. \textbf{(d)} A logical $\text{CZ}$ gate applied on logical qubit A and B on the state $\ket{\text{++++}}$. The blue bar represents the experimental measurement results, whereas the orange bar corresponds to the simulation data. The error bars are computed from quantum projection noise and represent one standard deviation of statistical uncertainty.}
    \label{fig:SS_eigenvalues_happy_init}
\end{figure}

\section{Syndrome table for the pentagon code}\label{SS:Lookup_table_HaPPY}
We investigate if an induced error produces the syndrome predicted by theory in order to demonstrate the ability of the pentagon code to detect and successively correct for single qubit errors. The syndromes derived from theory, along with their corresponding errors, are presented in Tab.~\ref{tab:SS:Syndrome_HaPPY}. To verify the table experimentally, single qubit errors corresponding to the gray highlighted syndromes are introduced in the experiment and compared to the theory. The experimental results are given in Fig.~\ref{fig:Syndrome}. The blue bars correspond to the experiment while the orange bars are numerical simulations, both agreeing with each other and the theory. 
\begin{figure}[th!]
    \centering
    \includegraphics[width=\linewidth]{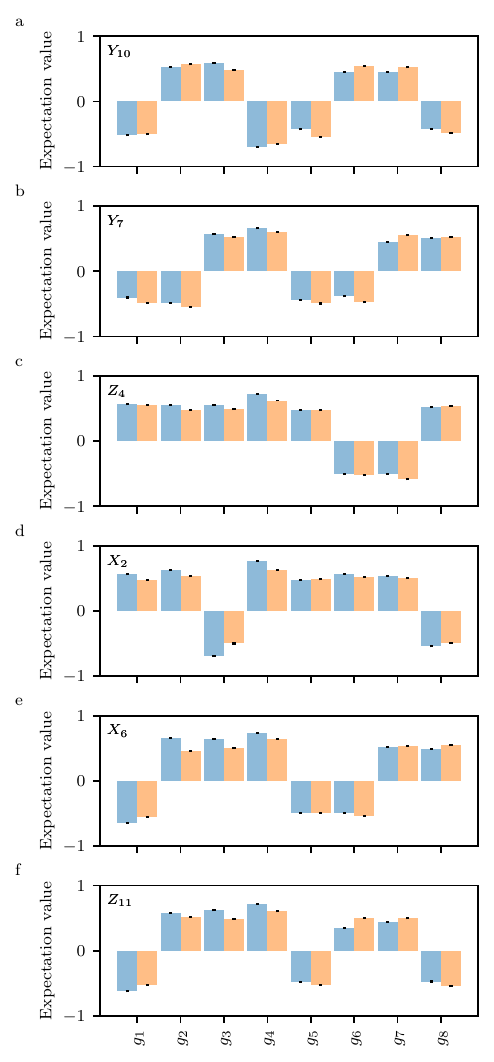}
    \caption{\textbf{Syndrome extraction.} \textbf{(a-f)} Extracted syndrome for the induced error given in the top left corner of each subfigure. The light blue bars represent the measured values, while the orange bars correspond to the simulation results. The error bars are computed from quantum projection noise and represent one standard deviation of statistical uncertainty. The measured syndromes are in agreement with the theoretical syndromes highlighted in gray in Tab.~\ref{tab:SS:Syndrome_HaPPY}.} 
    \label{fig:Syndrome}
\end{figure}

\begin{table}[ht]
\caption{The corresponding syndrome for a single $X$-, $Y$-, or $Z$-error on a physical qubit is presented. The simulated and measured pairs in Fig.~\ref{fig:Syndrome} are highlighted in gray.}
\begin{tabular}{c|c|c}
\label{tab:SS:Syndrome_HaPPY}
\textbf{Qubit}\phantom{-} & \phantom{-}\textbf{Error}\phantom{-} & \textbf{Syndrome $g_1-g_8$} \\
\hline
$0$ & X & \phantom{-}1, \phantom{-}1, -1, \phantom{-}1, \phantom{-}1, \phantom{-}1, -1, -1 \\
$0$ & Y & \phantom{-}1, \phantom{-}1, -1, -1, \phantom{-}1, \phantom{-}1, \phantom{-}1, -1 \\
$0$ & Z & \phantom{-}1, \phantom{-}1, \phantom{-}1, -1, \phantom{-}1, \phantom{-}1, -1, \phantom{-}1 \\
$1$ & X & \phantom{-}1, \phantom{-}1, -1, -1, \phantom{-}1, \phantom{-}1, \phantom{-}1, \phantom{-}1 \\
$1$ & Y & \phantom{-}1, \phantom{-}1, -1, -1, \phantom{-}1, \phantom{-}1, -1, -1 \\
$1$ & Z & \phantom{-}1, \phantom{-}1, \phantom{-}1, \phantom{-}1, \phantom{-}1, \phantom{-}1, -1, -1 \\
\cellcolor{gray!30}$2$ &\cellcolor{gray!30} X &\cellcolor{gray!30} \phantom{-}1, \phantom{-}1, -1, \phantom{-}1, \phantom{-}1, \phantom{-}1, \phantom{-}1, -1 \\
$2$ & Y & \phantom{-}1, \phantom{-}1, -1, -1, \phantom{-}1, \phantom{-}1, -1, \phantom{-}1 \\
$2$ & Z & \phantom{-}1, \phantom{-}1, \phantom{-}1, -1, \phantom{-}1, \phantom{-}1, -1, -1 \\
$3$ & X & \phantom{-}1, -1, \phantom{-}1, \phantom{-}1, \phantom{-}1, -1, -1, \phantom{-}1 \\
$3$ & Y & \phantom{-}1, -1, -1, \phantom{-}1, \phantom{-}1, \phantom{-}1, -1, \phantom{-}1 \\
$3$ & Z & \phantom{-}1, \phantom{-}1, -1, \phantom{-}1, \phantom{-}1, -1, \phantom{-}1, \phantom{-}1 \\
$4$ & X & \phantom{-}1, -1, -1, \phantom{-}1, \phantom{-}1, \phantom{-}1, \phantom{-}1, \phantom{-}1 \\
$4$ & Y & \phantom{-}1, -1, -1, \phantom{-}1, \phantom{-}1, -1, -1, \phantom{-}1 \\
$4$ \cellcolor{gray!30}& Z \cellcolor{gray!30}&\cellcolor{gray!30} \phantom{-}1, \phantom{-}1, \phantom{-}1, \phantom{-}1, \phantom{-}1, -1, -1, \phantom{-}1 \\
$5$ & X & \phantom{-}1, -1, \phantom{-}1, \phantom{-}1, \phantom{-}1, \phantom{-}1, -1, \phantom{-}1 \\
$5$ & Y & \phantom{-}1, -1, -1, \phantom{-}1, \phantom{-}1, -1, \phantom{-}1, \phantom{-}1 \\
$5$ & Z & \phantom{-}1, \phantom{-}1, -1, \phantom{-}1, \phantom{-}1, -1, -1, \phantom{-}1 \\
$6$ \cellcolor{gray!30}& X \cellcolor{gray!30}&\cellcolor{gray!30} -1, \phantom{-}1, \phantom{-}1, \phantom{-}1, -1, -1, \phantom{-}1, \phantom{-}1 \\
$6$ & Y & -1, -1, \phantom{-}1, \phantom{-}1, \phantom{-}1, -1, \phantom{-}1, \phantom{-}1 \\
$6$ & Z & \phantom{-}1, -1, \phantom{-}1, \phantom{-}1, -1, \phantom{-}1, \phantom{-}1, \phantom{-}1 \\
$7$ & X & -1, -1, \phantom{-}1, \phantom{-}1, \phantom{-}1, \phantom{-}1, \phantom{-}1, \phantom{-}1 \\
\cellcolor{gray!30}$7$&\cellcolor{gray!30} Y &\cellcolor{gray!30} -1, -1, \phantom{-}1, \phantom{-}1, -1, -1, \phantom{-}1, \phantom{-}1 \\
$7$ & Z & \phantom{-}1, \phantom{-}1, \phantom{-}1, \phantom{-}1, -1, -1, \phantom{-}1, \phantom{-}1 \\
$8$ & X & -1, \phantom{-}1, \phantom{-}1, \phantom{-}1, \phantom{-}1, -1, \phantom{-}1, \phantom{-}1 \\
$8$ & Y & -1, -1, \phantom{-}1, \phantom{-}1, -1, \phantom{-}1, \phantom{-}1, \phantom{-}1 \\
$8$ & Z & \phantom{-}1, -1, \phantom{-}1, \phantom{-}1, -1, -1, \phantom{-}1, \phantom{-}1 \\
$9$ & X & \phantom{-}1, \phantom{-}1, \phantom{-}1, -1, -1, \phantom{-}1, \phantom{-}1, -1 \\
$9$ & Y & -1, \phantom{-}1, \phantom{-}1, -1, -1, \phantom{-}1, \phantom{-}1, \phantom{-}1 \\
$9$ & Z & -1, \phantom{-}1, \phantom{-}1, \phantom{-}1, \phantom{-}1, \phantom{-}1, \phantom{-}1, -1 \\
$10$ & X & -1, \phantom{-}1, \phantom{-}1, -1, \phantom{-}1, \phantom{-}1, \phantom{-}1, \phantom{-}1 \\
$10$ \cellcolor{gray!30}& \cellcolor{gray!30}Y &\cellcolor{gray!30} -1, \phantom{-}1, \phantom{-}1, -1, -1, \phantom{-}1, \phantom{-}1, -1 \\
$10$ & Z & \phantom{-}1, \phantom{-}1, \phantom{-}1, \phantom{-}1, -1, \phantom{-}1, \phantom{-}1, -1 \\
$11$ & X & \phantom{-}1, \phantom{-}1, \phantom{-}1, -1, -1, \phantom{-}1, \phantom{-}1, \phantom{-}1 \\
$11$ & Y & -1, \phantom{-}1, \phantom{-}1, -1, \phantom{-}1, \phantom{-}1, \phantom{-}1, -1 \\
\cellcolor{gray!30}$11$ &\cellcolor{gray!30} Z &\cellcolor{gray!30} -1, \phantom{-}1, \phantom{-}1, \phantom{-}1, -1, \phantom{-}1, \phantom{-}1, -1 \\
\end{tabular}
\end{table}

\section{General encoding for graph codes without extra qubits}\label{sect:enc_extra_q}

The encoding and partial decoding method proposed in Ref.~\cite{Angl_s_Munn__2024} requires $k+n$ qubits. 
Although this method was convenient due to its graphical interpretation, it requires $k$ additional qubits that are initially situated in the bulk.
Here we show a different way to encode through a graph code that obviates this need. 
In this scenario, the $k$ bulk qubits are part of the $n$ boundary qubits. After applying an encoding unitary operator, the resulting state occupies the entire boundary. 
This method will help us to perform a more efficient partial decoding for our specific case.

To start, define first a graph code $G$ that encodes a $k$-qubit Hilbert space $\HH_{\text{in}}$ to the subspace $\CC$ from a $n$-qubit Hilbert space $\HH_{\text{out}}$.
Let $\{g_i\}^{n-k}_{i=1}$ be the generators of this code subspace $\CC$ written in its \emph{standard form},
\begin{equation}\label{eq:codesub}
        \CC_{G}=
        \begin{pmatrix}[cc|cc|c]
                \one_{n-k} & \nothing{A} &\nothing{B} & \nothing{C} & \omega^{\intercal}
            \end{pmatrix}\,.
    \end{equation}
The above equation is expressed in the check matrix representation of the stabilizer formalism \cite{nielsen_chuang_2010} where the rows correspond to the $n-k$ generators.
Here $A$ and $C$ are $(n-k) \times k$ matrices, $B$ is a $(n-k) \times (n-k)$ square matrix whose diagonal is zero, 
$\one_{n-k}$ is the identity matrix of size $n-k$ and $\omega$ is a real vector of size $n-k$.
      
In this formalism, the X-part and Z-part of the generators are represented by the binary matrices $(\one_{n-k} \,\, A)$ and $(B \,\, C)$ respectively.
The vector $\omega$ contains the phases of the generators: 
$\omega _i=0$ if the sign of the generator $g_i$ is positive
and $\omega_i=1$ if it is negative.
We note that elementary row operations modulo 2 in this representation do not change the stabilizer space.
In contrast, performing elementary column operations modulo 2 is equivalent to applying the following gates to the generators~\cite{Gheorghiu_2014}:
\begin{enumerate}[(i)]
    \item $\CNOT_{ab}$ gates corresponds to sum column $a$ to column $b$ for the $X$-part and sum column $b$ to column $a$ for the $Z$-part,
    \item $\text{SWAP}_{ab}$ gates corresponds to interchanging $a$ and $b$ columns.
\end{enumerate}
Elementary row operators induce {\em non-trivial phases} to the vector phase $\omega$ as a consequence of the Pauli relations. Such phases only depend on the X- and Z-parts of the check matrix: For example, summing the rows $(11|00|0)$ and
$(00|11|0)$ yields $(11|11|1)$, corresponding to
$(X\ot X)(Z\ot Z)=-(Y\ot Y)$. On the contrary, we refer to changes in the vector phase $\omega$ that arise from adding its elements as {\em trivial phases}.
In a similar way, applying $\CNOT_{ab}$ also changes the vector phase non-trivially~\cite{PhysRevA.70.052328}.

\bigskip
\noindent {\bf Encoding by projection.} 
We use the following property shown in~Ref.~\cite[Section 3.1]{10.1007/978-3-642-20901-7_9}
:
Recall that in standard form, the $i$th generator contains a Pauli $X$ on the coordinate $i$. Thus, we can write
$g_i = X_i \ot \tilde g_i$.
Its action on a state of the form $
\ket{\psi} = \ket{0}_i \ot \ket{\tilde \psi}$ is then given by a Hadamard and controlled-gate operation
\begin{equation}\label{eq:single_project}
\tfrac{1}{2}(\one + g_i) \ket{\psi} = \frac{\text{controlled-}\tilde g_i}{\sqrt{2}}
H_i (\ket{0}_i \ot \ket{\tilde\psi})\,.
\end{equation}
Now note that $\tilde g_i$ contains only Pauli $Z$ gates on the first $n-k$ qubits. 
Thus, they act trivially on $\ket{0}^{\otimes n-k}$, whereas the Hadamard $H$ only affects coordinate $i$.

As a consequence,
projecting a $\ket{\phi_{\text{in}}}\in \HH_{\text{in}}$
onto the code subspace $P_\CC$ yields,
\begin{equation}\label{eq:PUre}
        P_\CC \big(\ket{0}^{\otimes n-k} \otimes \ket{\phi_{\text{in}}}\big) =  \frac{U_\CC}{2^{(n-k)/2}}\big(\ket{0}^{\otimes n-k} \otimes \ket{\phi_{\text{in}}}\big)  \,,
    \end{equation}
where $U_\CC$ is a sequence of Hadamard and controlled Pauli gates given by Eq.~\eqref{eq:single_project}:
\begin{align}
U_\CC = \prod^{n-k}_{i=1} ({\text{controlled-}\tilde g_i}) H_i.
\end{align}
By Eq.~\eqref{eq:PUre}, 
\begin{equation}\label{eq:newlogbasis}
  \ket{U_{i_1 \dots  i_k}} =
  U_\CC \big(\ket{0}^{\otimes n-k}\otimes \ket{i_1\dots i_k}\big)
    \end{equation}
spans an orthonormal basis of $\CC$.
In Proposition~\ref{lem:gates2} we explain how to obtain a
logical {\em graph} state basis $\ket{G_{i_1 \dots i_k}}$, which 
respect the action of logical operators $\Zbar$ 
extracted from the graph.
In general, the two bases are different
$\ket{U_{i_1 \dots i_k}} \neq \ket{G_{i_1 \dots i_k}}$.

\bigskip
\noindent {\bf Encoding with given $\Zbar$.} 
In practice, it is useful to have a basis that respects the choice of logical $\Zbar$ operators extracted from the graph code~\cite[Eq.~(14)]{Angl_s_Munn__2024}:
\begin{equation}\label{eq:Zlogical}
                \Zbar_{}=
                \begin{pmatrix}[cc|cc]
                        \mathbb{0} & \nothing{D} &\nothing{F} & \nothing{G} \,|\,\gamma^\intercal
                    \end{pmatrix} \quad \text{where} \quad  \oper{rank}(\nothing{D})=k.
\end{equation}
In the above equation, we expressed the operators in the check matrix representation, where the $k$ rows correspond to the logical operators
$\Zbar$. The matrix $\mathbb{0}$ is a $k \times (n-k)$ matrix of all zero entries, $F$ is a $k \times (n-k)$ matrix, $D$ and $G$ are of size $k\times k$ and $\gamma$ is a vector phase of size $k$. 
Since $\Zbar$ is derived from a graph code, performing row operations on $(\nothing{D}|\nothing{G}|\gamma^\intercal)$ results in $(\one_k | \Gamma_{\Zbar}|0^\intercal)$, where $\Gamma_{\Zbar}$ has the structure of an adjacency matrix with zeros on the diagonal, defining a set of edges $E_{\Zbar}$. 

Let 
\begin{align}\label{eq:encwoanc1}
V = 
 \left( \prod_{(u,v)\in E_{\Zbar}} \CZ_{uv} \right)  U_{\nothing{D}} H^{\ot k}\,,
\end{align}
which we use to preprocess the input,
$    \ket{\widetilde{\phi_{\text{in}}}} = 
    V
    \ket{\phi_{\text{in}}}\,.$

Here, $U_D$ is a unitary matrix that transforms the $k$-qubit state with generators $(\one_k| 0)$ to another with generators $(D |0 )$.
Recall that matrix $D$ in Eq.~\eqref{eq:newlogbasis} has rank $k$. Thus, it can be obtained by elementary column operations from the identity matrix $\one$.  The unitary matrix $U_D$ is then decomposed as a sequence of gates \text{SWAP} and \text{CNOT}. 

\begin{proposition}\label{lem:gates2}
    Let $\CC_{G}$ be the code subspace written in the standard form 
    and let
    $\Zbar$ be a set of logical-$Z$ operators
    given by Eq.~\eqref{eq:Zlogical}. 
    Then,
    \begin{equation}\label{eq:encwoanc0}
        U_\CC \big(
        \ket{0}^{\otimes n-k} \otimes V \ket{\phi_{\text{in}}}
        \big)
        =
        \ket{\phi_{\text{enc}}}
    \end{equation}
        encodes $\ket{\phi_{\text{in}}}$ such that 
        the $\Zbar$ act as logical-$Z$ operators.
    \end{proposition}
\begin{proof}
    For the relation Eq.~\eqref{eq:encwoanc0} to hold, 
    we need to show that $V$ of Eq.~\eqref{eq:encwoanc1} is 
    the correct k-qubit unitary acting on the computational basis, such that $\ket{\mu_{i_1 \dots i_k}} = V\ket{i_1 \dots i_k}$ is mapped to 
    $\ket{G_{i_1 \dots i_k}}$ by the unitary $U_\CC$.
    
    Suppose the $\ket{\mu_{i_1\dots i_k}}$ form an orthonormal basis so that
        \begin{equation}\label{eq:newdem-3}
                U_\CC \big(\ket{0}^{\otimes n-k} \otimes \ket{\mu_{i_1\dots i_k}}\big)=\ket{G_{i_1\dots i_k}}\,.
            \end{equation}
        Our goal is to find a finite set of unitary gates mapping each $\ket{i_1\dots i_k}$ to $\ket{\mu_{i_1\dots i_k}}$.     
        To achieve that, we first derive the stabilizers of  $\ket{\mu_{i_1\dots i_k}}$. 
        Denote by $\{ \Zbar_1,\dots,\Zbar_k\}$
        the logical operators from Eq.~\eqref{eq:Zlogical}.
        Using Eq.~\eqref{eq:PUre} and Eq.~\eqref{eq:newdem-3}, we obtain
        \begin{align}\label{eq:newdem-2}
        \ket{G_{i_1\dots i_k}} &= (-1)^{i_\ell}\Zbar_\ell \ket{G_{i_1\dots i_k}} \nn \\
        &
        = (-1)^{i_\ell}\Zbar_\ell U_\CC 
        (\ket{0}^{\otimes n-k} \otimes \ket{\mu_{i_1\dots i_k}}) \\
         &= 
         (-1)^{i_\ell} \Zbar_\ell  2^{(n-k)/2}{P_\CC}                  
         (\ket{0}^{\otimes n-k} \otimes \ket{\mu_{i_1\dots i_k}}) \nn\\
         &= 
         (-1)^{i_\ell} 2^{(n-k)/2}{P_\CC} \Zbar_\ell (\ket{0}^{\otimes n-k} \otimes \ket{\mu_{i_1\dots i_k}}) \,,\nn
        \end{align}
        for all $\ell$ and $i_1, \dots, i_k$. 
        Here we used that $P_\CC$ commutes with all $\Zbar_\ell$ and 
        $(-1)^{i_\ell}\Zbar_\ell \ket{G_{i_1\dots i_k}}=\ket{G_{i_1\dots i_k}}$.
        
        The first $n-k$ qubits are initialized in $\ket{0}$, 
        and by Eq.~\eqref{eq:Zlogical}, $\Zbar_\ell$ acts trivially on the first $n-k$ qubits, and we can define an effective action 
        $\Zbar^{\text{eff}}_\ell$ on the last $k$ qubits. 
        We now continue Eq.~\eqref{eq:newdem-2} by using $\Zbar^{\text{eff}}_\ell$ and again Eq.~\eqref{eq:PUre} which leads to 
        \begin{align}\label{eq:extrademo}
        \ket{G_{i_1\dots i_k}}
         &=(-1)^{i_\ell} 2^{(n-k)/2} P_\CC (\ket{0}^{\otimes n-k} \otimes \Zbar^{\text{eff}}_\ell \ket{\mu_{i_1\dots i_k}}) \nn \\
        &= (-1)^{i_\ell} U_\CC (\ket{0}^{\otimes n-k} \otimes \Zbar^{\text{eff}}_\ell \ket{\mu_{i_1\dots i_k}}) \nn \\
        &= (-1)^{i_\ell} U_\CC \Zbar_\ell  (\ket{0}^{\otimes n-k} \otimes \ket{\mu_{i_1\dots i_k}}) \,.
        \end{align}
         
        Comparing the Eq.~\eqref{eq:extrademo} with Eq.~\eqref{eq:newdem-3},
        one concludes that $\{(-1)^{i_\ell}\Zbar_\ell\}$ stabilizes $\ket{0}^{\otimes n-k} \otimes \ket{\mu_{i_1\dots i_k}}$.
        In particular, by Eq.~\eqref{eq:Zlogical} 
        and noting that the Pauli elements from $F$ acts trivially on $\ket{0}^{\ot n-k}$, 
        the stabilizers of  
        $\ket{\mu_{i_1\dots i_k}}$ are
        \begin{equation}\label{eq:newdem-1}
                \begin{pmatrix}[c|c|c] \nothing{D} & \nothing{G} & \omega^{\intercal}+\gamma^\intercal \end{pmatrix}  \quad\text{with}\quad \omega=(i_1,\dots,i_k)\,,
            \end{equation} 
        where $\omega$ is a phase vector 
        which determines the basis element,
        and $\gamma$ is the constant vector phase from the $\Zbar$ operators.
        
        We continue the proof by deriving the gates to map the stabilizers from $\ket{i_1\dots i_k}$ to the ones from $\ket{\mu_{i_1\dots i_k}}$.
        We first start by writing the check matrix of the state $H^{\otimes k}\ket{i_1 \dots i_k}$ as
    \begin{equation}\label{eq:newdem0}
           \begin{pmatrix}[c|c|c] \one_k & 0 & \omega^{\intercal} \end{pmatrix}\,.
        \end{equation}
    
       By elementary
       {\em column} operations to the above equation, we obtain
        \begin{equation}\label{eq:newdem0_}
                \begin{pmatrix}[c|c|c] \nothing{D}  & 0 & \omega^{\intercal}\end{pmatrix} \,.
        \end{equation}
        This corresponds to apply an unitary operator $U_{\nothing{D}}$ which can be decomposed by a set of \text{SWAP} and $\CNOT$ gates. Thus, Eq.~\eqref{eq:newdem0_} are the generators from $U_D H^{\otimes k}\ket{i_1 \dots i_k}$.
        Note since the Z-part from Eq.~\eqref{eq:newdem0} is zero, the phase vector does not change (see Ref.~\cite{PhysRevA.70.052328}).
        
        We proceed by performing {\em row} operations to Eq.~\eqref{eq:newdem0_} 
        so that $D$ is mapped to the identity,
        \begin{equation}\label{eq:newdem2}
                \begin{pmatrix}[c|c|c] \one_k & 0 & \widetilde{\omega}^{\intercal} \end{pmatrix}\,,
        \end{equation} 
        Here $\widetilde\omega^{\intercal}=D^{-1}\omega^{\intercal}$, since the Z part of the check matrix is zero and, thus, the change from $\omega$ to $\widetilde\omega$ only contains trivial phases.
        Recall that such operators do not change the stabilizer space and, therefore, Eq.~\eqref{eq:newdem2} are still generators of $U_D H^{\otimes k}\ket{i_1 \dots i_k}$.
        
        Finally, we apply a set of \text{CZ} gates:
        Let $E_{\Zbar}$ be the edges determined by $\Gamma_{\Zbar}$ [see below Eq.~\eqref{eq:Zlogical}] and apply $\prod_{(u,v)\in E_{\Zbar}} \CZ_{uv}$ to Eq.~\eqref{eq:newdem2} leading to
        \begin{equation}\label{eq:newdem4}
                \begin{pmatrix}[c|c|c] \one_k & \Gamma_{\Zbar} & \widetilde{\omega}^{\intercal} \end{pmatrix}\,.
        \end{equation}
        Note that $\tilde \omega$ does not change, due to
        \begin{align}
            \operatorname{CZ} (X\ot I) \operatorname{CZ} &= X\ot Z \,,\nn\\
            \operatorname{CZ} (I\ot X) \operatorname{CZ}& = Z \ot X\,,
        \end{align}
        and all $\operatorname{CZ}$ operator commute with $Z$'s.
        
        Recall that performing row operations on $(\nothing{D}|\nothing{G}|\gamma^\intercal)$ results in $(\one_k | \Gamma_{\Zbar}|0^\intercal)$ [also see below Eq.~\eqref{eq:Zlogical}]. Therefore, doing the row operations in reverse on Eq.~\eqref{eq:newdem4} leads to
        \begin{align}\label{eq:newdem3}
       (\nothing{D}\,|\,\nothing{G}\,|\,\omega^{\intercal}+\gamma^\intercal) \,,
        \end{align}
    which are the stabilizers of $\ket{\mu_{i_1\dots i_k}}$, as shown  in Eq.~\eqref{eq:newdem-1}.
    Here $\gamma$ are the non-trivial phases appearing because of the row operations, and the change $\omega^{\intercal}=D\widetilde\omega^{\intercal}$ only contains trivial phases.
        Following the unitary gates applied to transform Eq.~\eqref{eq:newdem0} to~\eqref{eq:newdem3}, we conclude that
        \begin{align}\label{eq:newdem5}
                \ket{\mu_{i_1\dots i_k}} &=  \left( \prod_{(u,v)\in E_{\Zbar}} \CZ_{uv} \right)  U_{\nothing{D}}H^{\otimes k}\ket{i_1 \dots i_k} \,.
        \end{align}
           By Eq.~\eqref{eq:newdem-3} and Eq.~\eqref{eq:newdem5}, we obtain Eq.~\eqref{eq:encwoanc0} with $V$ defined in Eq.~\eqref{eq:encwoanc1}.
        This ends the proof.
    \end{proof}
Note that Proposition~\ref{lem:gates2} applies for a given code subspace $\CC_G$ and logical operators $\Zbar$, and both do not depend on the interactions between the qubits from $\HH_{\text{in}}$.
If these qubits interact non-trivially in the graph code, we need to take into account new phases appearing in Eq.~\eqref{eq:newdem5}.
Denote by $E_{\text{in}}$ the edges between the qubits of $\HH_{\text{in}}$.
Similarly to the encoding method from Ref.~\cite{Angl_s_Munn__2024}, Eq.~\eqref{eq:encwoanc1} is then modified as follows
     \begin{align}\label{eq:encwoanc_graph}
                   V=&  \left( \prod_{(u,v)\in E_{\Zbar}} \CZ_{uv} \right)  U_{\nothing{D}}H^{\otimes k} \nn \\  &\cdot \left( \prod_{(u,v)\in E_{\text{in}}} \CZ_{uv} \right)  \,.
    \end{align}

\noindent {\bf Application of Proposition~\ref{lem:gates2}:} The above proposition can be used to extract the encoding unitary from any graph code.
Here, we are using it to extract from the AME state of $6$ qubits illustrated in Fig.~\ref{fig:AME}(a), the unitary operator that maps $3$ qubits from the bottom to the $3$ qubits from the top [see Fig.~\ref{fig:AME}(b)]. 
In this particular case, the code subspace is trivial $\CC_G=\emptyset$ since $n=k=3$. 
Thus, to apply such a proposition, we only need to obtain the logical operators $\Zbar$ from the graph code.
We first write the check matrix for the graph state illustrated in Fig.~\ref{fig:AME}(a) as follows,
\begin{equation}\label{eq:checkAME}
    \begin{pmatrix}[@{\hskip 0.2cm}ccc c ccc @{\hskip 0.2cm}| @{\hskip 0.2cm} ccc c ccc@{\hskip 0.2cm} ]
        1 & 0 & 0  
        && 0 & 0 & 0 &
        0 & 1 & 1 
        && 1 & 1 & 0 \\
        %
        0 & 1 & 0 
        && 0 & 0 & 0 &
        1 & 0 & 1 
        && 0 & 1 & 0 \\
        %
        0 & 0 & 1 
        && 0 & 0 & 0 &
        1 & 1 & 0
        && 0 & 1 & 1 \\
        \\
        0 & 0 & 0 
        && 1 & 0 & 0 &
        1 & 0 & 0 
        && 0 & 1 & 0 \\
        %
        0 & 0 & 0 
        && 0 & 1 & 0 &
        1 & 1 & 1 
        && 1 & 0 & 1 \\
        %
        0 & 0 & 0 
        && 0 & 0 & 1 &
        0 & 0 & 1 
        && 0 & 1 & 0 
    \end{pmatrix}\,,
\end{equation}
where the first three columns corresponds to the three top qubits ($n=3$) and the last three to the bottom ones ($k=3$).
As shown in Ref.~\cite[Eq.(14)]{Angl_s_Munn__2024}, the logical operators $\Xbar$ are derived by removing the columns that correspond to the qubits in $\HH_{\text{in}}$ leading to
\begin{equation}\label{eq:RX}
   \begin{pmatrix} R \\ \\ \Xbar\end{pmatrix} = \begin{pmatrix}[@{\hskip 0.2cm} ccc @{\hskip 0.2cm}| @{\hskip 0.2cm} ccc@{\hskip 0.2cm} ]
        1 & 0 & 0  
        &
        0 & 1 & 0 
        \\
        %
        0 & 1 & 0 
        &
        1 & 0 & 1 
         \\
        %
        0 & 0 & 1 
        &
        0 & 1 & 0 
         \\
        \\
        0 & 0 & 0 
        &
        1 & 0 & 0 
        \\
        %
        0 & 0 & 0 
        &
        1 & 1 & 1 
       \\
        %
        0 & 0 & 0 
        &
        0 & 0 & 1  
    \end{pmatrix}\,,
\end{equation}
where $R$ is the upper part of the matrix, and $\Xbar$ the logical operators.
Performing now row operations on $R$ leads to logical operators $\Zbar$,
\begin{equation}
   \begin{pmatrix} \Zbar \\ \\ \Xbar\end{pmatrix} = \begin{pmatrix}[@{\hskip 0.2cm} ccc @{\hskip 0.2cm}| @{\hskip 0.2cm} ccc@{\hskip 0.2cm} ]
        1 & 1 & 0  
        &
        1 & 1 & 1 
        \\
        %
        0 & 1 & 0 
        &
        1 & 0 & 1 
         \\
        %
        0 & 1 & 1 
        &
        1 & 1 & 1 
         \\
        \\
        0 & 0 & 0 
        &
        1 & 0 & 0 
        \\
        %
        0 & 0 & 0 
        &
        1 & 1 & 1 
       \\
        %
        0 & 0 & 0 
        &
        0 & 0 & 1  
    \end{pmatrix}\,.
\end{equation}
The row operations are performed so that $\Xbar$ and $\Zbar$ satisfy the expected Pauli commutation relations. 
Using Eq.~\eqref{eq:Zlogical} as a reference, we then write
\begin{align}
D= \begin{pmatrix}
    1 & 1 & 0 \\
    0 & 1 & 1 \\
    0 & 1 & 1 \\
\end{pmatrix}\quad \text{and} \quad G=\begin{pmatrix}
    1 & 1 & 1 \\
    1 & 0 & 1 \\
    1 & 1 & 1 \\
\end{pmatrix}\,.
 \end{align}
We now use Proposition~\ref{lem:gates2} to derive the encoding gates. Since $\CC_G=\emptyset$ and the qubits in $\HH_{\text{in}}$ interact non-trivially, the encoding operators are determined by Eq.~\eqref{eq:encwoanc_graph}:
\begin{enumerate}
    \item From Eq.~\eqref{eq:checkAME}, one can see that the interactions between the qubits in $\HH_{\text{in}}$ are described by $E_{\text{in}}=\{ (12),(23), (13)\}$.
    \item One has the unitary $U_D= \CX_{12}\CX_{32}$, which transforms $(\one_k| 0) \rightarrow (D|0)$.
    \item The Z-part (the three right columns) of the rows in R in Eq.~\eqref{eq:RX} corresponds to $\Gamma_{\Zbar}$, 
    since
    performing row operations on $(\nothing{D}|\nothing{G}|0^\intercal)$ results in $(\one_k | \Gamma_{\Zbar}| 0^\intercal)$. As a consequence, $E_{\Zbar}=\{(12),(23)\}$.
\end{enumerate}
The encoding gates are also illustrated in Fig.~\ref{fig:unitary_rec}(c).

	

\begin{figure}
	\centering
	
	\begin{subfigure}[t]{0.25\linewidth}
		\captionsetup{justification=raggedright,
			singlelinecheck=false,
			position=top}
		\caption{{\small a}}
		\includegraphics[width=\linewidth]{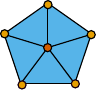}
	\end{subfigure}
	\hspace{2cm}
	\begin{subfigure}[t]{0.25\linewidth}
		\captionsetup{justification=raggedright,
			singlelinecheck=false,
			position=top}
		\caption{{\small b}}
		\includegraphics[width=\linewidth]{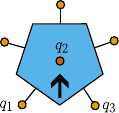}
	\end{subfigure}
	 \caption{\textbf{Unitary extraction from a graph code.} \textbf{(a)} A graph form of the absolutely maximally entangled state (AME) of 6 qubits.
		\textbf{(b)} This state represents a graph code, and it can be thought of as a unitary that goes from the three qubits of the bottom to the three qubits from the top. The arrow shows this direction.
        This unitary can be expressed as the unitary circuit shown in Fig.~\ref{fig:unitary_rec}(c).
	}
	\label{fig:AME}
\end{figure}
\section{Alternative partial decoding}\label{SS:alternative_partial_decoding}

In this paper, we perform a partial decoding procedure that differs from the one proposed in Ref.~\cite{Angl_s_Munn__2024}. 
This gives a reduced number of gates in the instance of Fig.~\ref{fig:MinimalGraph}(a), but we do not know whether this better scaling holds in general.
The method is based on Section~\ref{sect:enc_extra_q}. 
For simplicity, here we only consider the contraction of two pentagons. This corresponds to the tensor network that allows us to recover two bulk qubits from five boundary qubits, as shown in Fig.~\ref{fig:unitary_rec}(a). However, the same strategy applies to the other partial decoding strategies. 
Ref.~\cite{Angl_s_Munn__2024} first contracts two pentagons to then convert the resultant into a small graph code. In contrast, 
we first convert the two pentagons to a graph code that maps $3$ to $3$ qubits as shown in Fig.~\ref{fig:MinimalGraph}, and then contract them as shown in Fig.~\ref{fig:unitary_rec}.

Consider the absolutely maximally entangled state (AME) of 6 qubits in its graph state form, as shown in Fig.~\ref{fig:AME}(a).
Note that this AME state can be expressed as unitary that goes from $3$ to $3$ qubits written as
\begin{align}
    \ket{\text{AME}} &=\sum^1_{i_1,i_2,i_3=0} \ket{\phi_{i_1i_2i_3}}\ket{i_1i_2i_3} \quad \rightarrow \quad \nn \\ U  &=\sum^1_{i_1, i_2, i_3=0}\ketbra{\phi_{i_1 i_2 i_3}}{i_1i_2i_3}
\end{align}
and illustrated in Fig.~\ref{fig:AME}(b).
This unitary is derived using the method from the encoding described in Section~\ref{sect:enc_extra_q} for the case where $n=k=3$. As a result, the unitary $U$ is decomposed by the gates described in Fig.~\ref{fig:AME}(c).
Consider now the recovery cut from Fig.~\ref{fig:unitary_rec}(a). The unitary $U$ can then be concatenated as
\begin{align}
U^\dagger_h=(U^\dagger_{123} \ot \one_{45})(\one_{12} \ot U^\dagger_{345})
\end{align}
as illustrated in Fig.~\ref{fig:unitary_rec}(b)-(c).

As in Ref.~\cite{Angl_s_Munn__2024}, one needs to further modify $U^\dagger_h$ to partially decode a state that is encoded via the used graph code. This requires the addition of a Hadamard gate that leads to $\widetilde{U}^\dagger_h=U_h H_3$. The decoding circuit is shown in Fig.~\ref{fig:unitary_rec}(c).
The partial decoding operator requires $14$ two-qubit gates and $7$ one-qubit gates.
Similarly, one can derive the decoding circuit for two other different cuts as illustrated in Fig.~\ref{fig:D3}(c) and Fig.~\ref{fig:D4}(c). 

\newpage
\section{Partial decoding: Expectation values and measured decoding circuits}\label{SS:Partial_decoding_2_3}
Partial decoding is always performed by first initializing the code in $\ket{G^{}_{0000}}$ followed by the application of the partial decoding circuits shown in Fig.~\ref{fig:unitary_rec}, Fig.~\ref{fig:D3} or Fig.~\ref{fig:D4}. The unitary $U$ is defined in Fig.~\ref{fig:AME}(c). To decode the expectation value of a logical qubit, we apply $U$ to a designated boundary region of the code subspace, thereby mapping the logical qubit onto a single physical qubit. This qubit is measured and the expectation values are presented in Fig.~\ref{fig:exp_partial_decoding}(a)-(i). Subfigures (a)-(c) correspond to decoding logical qubits A-D without additional gates. Decoding 2 bulk qubits, 3 bulk qubits, and all four logical qubits yields

\vspace{0.2cm}
\begin{enumerate}[(i)]
    \item
    \begin{tabular}{l l}
    $A$: 0.616(15), & [sim: 0.436(16)] \\
    $B$: 0.296(14), & [sim: 0.422(16)] \\
    \end{tabular}

    \item
    \begin{tabular}{l l}
    $A$: 0.394(15), & [sim: 0.442(16)] \\
    $B$: 0.312(15), & [sim: 0.305(16)] \\
    $D$: 0.298(14), & [sim: 0.330(16)] \\
    \end{tabular}

    \item
    \begin{tabular}{l l}
    $A$: 0.600(15), & [sim: 0.441(16)] \\
    $B$: 0.142(11), & [sim: 0.103(9)] \\
    $C$: 0.252(14), & [sim: 0.227(13)] \\
    $D$: 0.472(15), & [sim: 0.323(16)]. \\
    \end{tabular}
\end{enumerate}
Subfigures (d)-(f) depict the same decoding process after applying an MS gate to qubits 6 and 10. Since these qubits are not involved in any decoding operation, the results should remain unaffected. We measure an expectation value for decoding 2, 3, and all 4 qubits of

\vspace{0.2cm}
\begin{enumerate}[(i)]
    \item
    \begin{tabular}{l l}
    $A$: 0.434(16), & [sim: 0.436(16)] \\
    $B$: 0.346(15), & [sim: 0.304(15)] \\
    \end{tabular}

    \item
    \begin{tabular}{l l}
    $A$: 0.446(16), & [sim: 0.436(16)] \\
    $B$: 0.274(14), & [sim: 0.297(14)] \\
    $D$: 0.300(15), & [sim: 0.322(15)] \\
    \end{tabular}

    \item
    \begin{tabular}{l l}
    $A$: 0.448(16), & [sim: 0.441(16)] \\
    $B$: 0.050(7), & [sim: 0.119(10)] \\
    $C$: 0.184(12), & [sim: 0.233(14)] \\
    $D$: 0.278(14), & [sim: 0.338(15)]. \\
    \end{tabular}
\end{enumerate}
Subfigures (g)-(i) show the decoding results after applying an MS gate on qubits 6 and 11. Since Qubit 11 is required to decode logical qubits B-D, this operation is expected to affect the results.  We measure a expectation value for decoding 2, 3, and 4 qubits respective
\vspace{0.2cm}
\begin{enumerate}[(i)]
    \item
    \begin{tabular}{l l}
    $A$: 0.402(16), & [sim: 0.451(16)] \\
    $B$: 0.312(15), & [sim: 0.298(15)] \\
    \end{tabular}

    \item
    \begin{tabular}{l l}
    $A$: 0.346(15), & [sim: 0.411(16)] \\
    $B$: 0.242(14), & [sim: 0.300(14)] \\
    $D$: 0.044(6),  & [sim: 0.022(5)] \\
    \end{tabular}

    \item
    \begin{tabular}{l l}
    $A$: 0.326(15), & [sim: 0.43(16)] \\
    $B$: 0.015(5),  & [sim: 0.007(3)] \\
    $C$: -0.016(4), & [sim: -0.014(4)] \\
    $D$: 0.100(9),  & [sim: -0.001(2)]. \\
    \end{tabular}
\end{enumerate}

\section{Partial decoding: Tomography}\label{SS:Partial_decoding_tomography}
We perform state tomography on the smaller subset of cuts $1-3$, namely performing five-qubit state tomography on the physical qubits used to decode cut 1 and 2 but performing four-qubit state tomography on the idle qubits for cut 3. We perform a maximum likelihood estimation to exclude non-physical results. Each state tomography is combined with a noise free simulation. The simulations are given in Fig.~\ref{fig:Tomography_data_all_cuts}(a), (c), and (e) and the reconstructed measurements are given in \ref{fig:Tomography_data_all_cuts}(b), (d), and (f). The reconstructed density matrices agree qualitatively with the simulations. 
\begin{figure*}
    \centering
    \includegraphics[width=\textwidth]{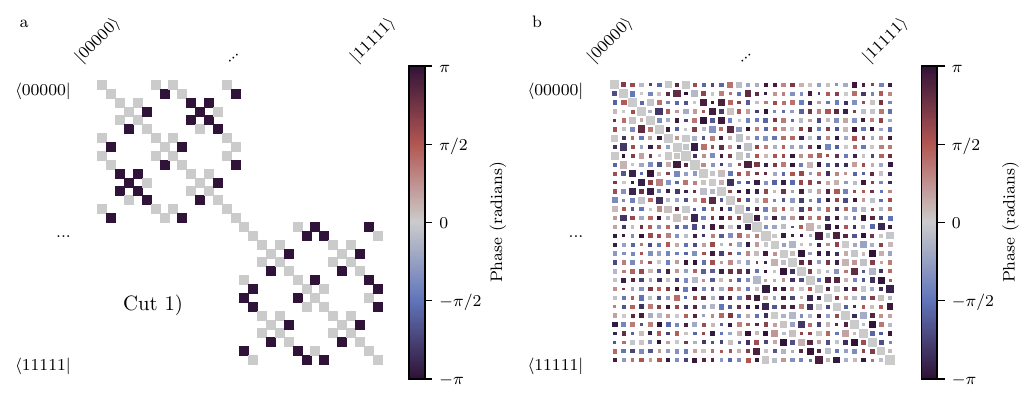}
    \includegraphics[width=\textwidth]{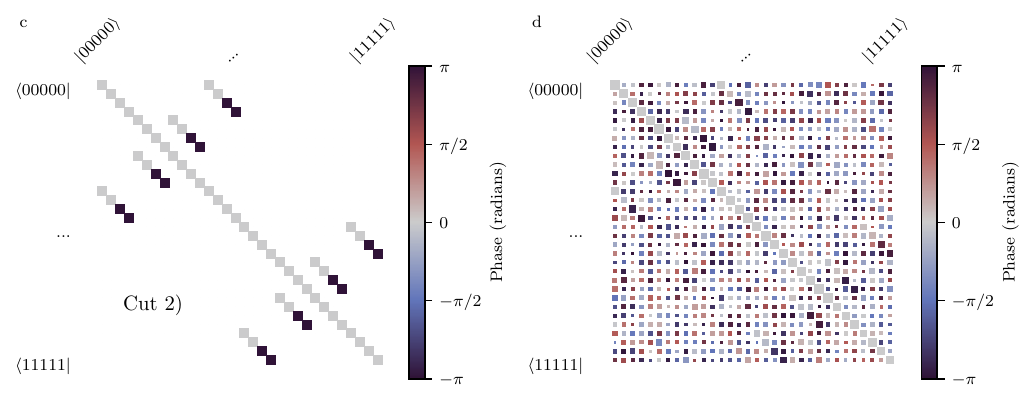}
    \includegraphics[width=\textwidth]{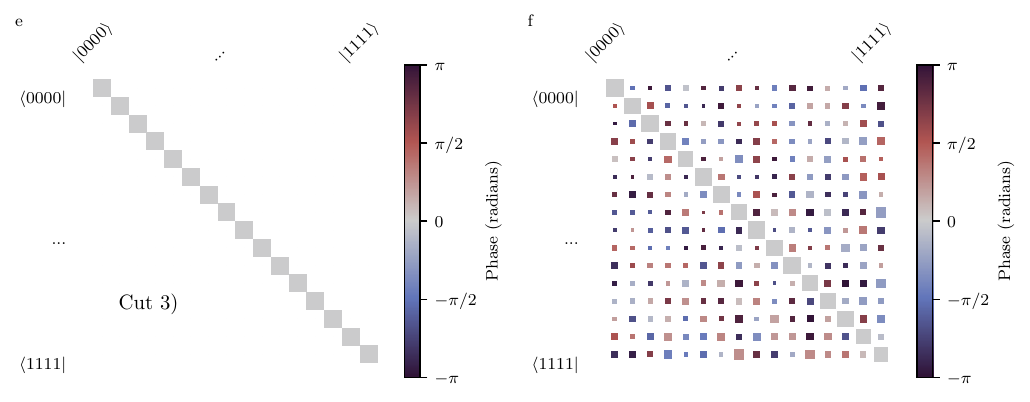}
    \caption{\textbf{State tomography on boundary qubits required for partial decoding.} \textbf{(a), (c), (d)} Simulated and \textbf{(b), (d), (f)} experimental state tomography data for cuts 1–3 are illustrated. The simulation is done in absence of noise. The experimental results reflect the general characteristics of the simulation. }
    \label{fig:Tomography_data_all_cuts}
\end{figure*}

\section{Logical operators, gates, and stabilizer generators of the holographic heptagon code}\label{SS:Stabilizer_generators_logical_operators_heptagon}
The logical operators of the holographic heptagon code are defined as follows:

\begin{tabbing}
    \hspace{1.5cm}\=$\overline{Z}_A=Z_{1_A}Z_{2_A}Z_{3_A}$, \hspace{0.7cm}\=$\overline{Z}_B=Z_{1_B}Z_{2_B}Z_{3_B}$,\\\\
    \>$\overline{X}_A=X_{1_A}X_{2_A}X_{3_A}$, 
    \>$\overline{X}_B=X_{1_B}X_{2_B}X_{3_B}$\\\\
    \>$\overline{Y}_A=Y_{1_A}Y_{2_A}Y_{3_A}$, 
    \>$\overline{Y}_B=Y_{1_B}Y_{2_B}Y_{3_B}$
\end{tabbing}

The stabilizer generators of the subspace are:
\begin{tabbing}
$g_1 = X_{1_A} X_{2_A} X_{4_A} X_{5_A}$, \hspace{0.7cm} \=$g_2=X_{2_A} X_{3_A} X_{5_A} X_{6_A}$,\\
$g_3 = X_{1_B} X_{2_B} X_{4_B} X_{5_B}$, \>$g_4 = X_{2_B} X_{3_B} X_{5_B} X_{6_B}$\\
$g_5 = X_{4_A} X_{5_A} X_{6_A} X_{4_B} X_{5_B} X_{6_B}$,\\\\
$g_6 = Z_{1_A} Z_{2_A} Z_{4_A} Z_{5_A}$, \>$g_7=Z_{2_A} Z_{3_A} Z_{5_A} Z_{6_A}$,\\
$g_8 = Z_{1_B} Z_{2_B} Z_{4_B} Z_{5_B}$, \>$g_9 = Z_{2_B} Z_{3_B} Z_{5_B} Z_{6_B}$\\
$g_{10} = Z_{4_A} Z_{5_A} Z_{6_A} Z_{4_B} Z_{5_B} Z_{6_B}$,
\end{tabbing}

The holographic heptagon code also features a "quasi"-transversal Hadamard gate. It is implemented by applying a physical Hadamard gate to each physical qubit encoding the logical qubit and a round of quantum error correction (QEC).
\begin{equation}
    \overline{H}_A = H_{1_A}H_{2_A}H_{3_A}H_{4_A}H_{5_A}H_{6_A} \quad +\quad\text{QEC}    
\end{equation}
We call the Hadamard gate "quasi" transversal since applying a Hadamard on logical qubit A induces an error on logical qubit B and vice versa. However, the induced error can be detected and compensated by an additional round of error correction. Therefore, stabilizer $g_5$ is mapped onto the auxiliary qubit $A_1$. To map the $X$ error onto the auxiliary qubit, $A_1$ is initially prepared in the $+$-basis by the means of a physical Hadamard gate. This is followed by a $\text{CNOT}$ gate on each qubit of the stabilizer, with $A_1$ as the control and the qubits of $g_5$ as the targets. An additional $H$ is applied to $A_1$ to rotate it into the $Z$-measurement basis. Similarly, $A_2$ is used to detect a negative expectation value of $g_{10}$. A $\text{CNOT}$ is applied on each qubit of stabilizer $g_{10}$, where the qubits of $g_{10}$ act as control and ancilla $A_2$ as target. The correction given in equation \ref{eq:Hadamard_correction_heptagon} is applied during postprocessing by the means of Pauli-frame update, depending on which flag was raised.\\
The transversal CNOT gate can be constructed from transversal single-qubit gates and a transversal MS gate, which is implemented using six physical MS gates
\begin{equation}
    \exp{\left(\text{i}\frac{\pi}{4}\overline{X}_A\otimes \overline{X}_B\right)}=\sum_{(i, j)\in (A, B)}\exp{\left(\text{i}\frac{\pi}{4}X_i\otimes X_j\right)}.
\end{equation}
Data showing the implementation of the logical zero state $\ket{00}_\text{hep}$, the quasi-transversal Hadamard gate and an MS gate on both logical qubits is shown in Fig.~\ref{fig:Heptagon_init}. The blue bars correspond to the measured data while the orange bars are the result from a numerical simulation.

\begin{figure}[th!]
    \centering
    \includegraphics[width=\linewidth]{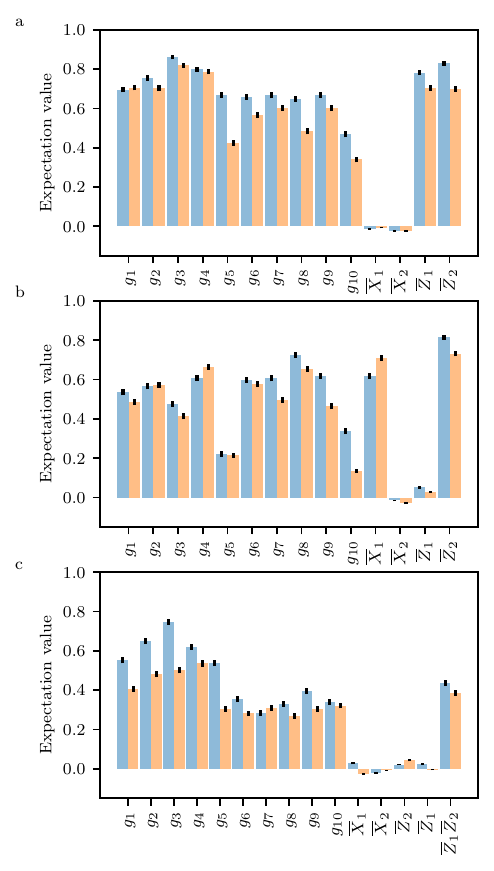}
    \caption{\textbf{State preparation of the heptagon code.}
    \textbf{(a-c)}
    Stabilizer generators and logical operators for the holographic heptagon code initialized in the logical states $\ket{\overline{00}}_\text{hep}$, 
    $\ket{\overline{+0}}_\text{hep}$, and 
    $(\ket{\overline{00}}_\text{hep}-i\ket{\overline{11}}_\text{hep})/\sqrt{2}$. 
    Blue bars correspond to experimental data while the orange bars correspond to numerical simulation. The error bars are computed from quantum projection noise and represent one standard deviation of statistical uncertainty. The implementation of the Hadamard gate on logical qubit $A$ is followed by post selecting the sign of stabilizer generators $g_5$ and $g_{10}$ and applying the corresponding correction.
    } 
    \label{fig:Heptagon_init}
\end{figure}

\end{document}